%% file: final_main.tex
\documentclass[11pt]{article}
\usepackage[T1]{fontenc}
\usepackage{lmodern,microtype}
\usepackage{inconsolata}
\usepackage{libertine}
\usepackage[margin=1in]{geometry}
\usepackage{amsmath,amssymb,amsthm,mathtools,bbm}
\usepackage{enumitem}
\usepackage{nicefrac}
\usepackage{thm-restate}
\usepackage{algorithm,algpseudocode}
\usepackage{todonotes}
\usepackage[utf8]{inputenc}
\usepackage{booktabs,tabularx,array}
\usepackage{xcolor}
\usepackage{tikz}
\usetikzlibrary{arrows.meta,calc,positioning}
\usepackage{float,needspace}

\usepackage[colorlinks=true,linkcolor=blue,citecolor=blue,urlcolor=blue,
            hypertexnames=false]{hyperref}
\usepackage[capitalise,nameinlink]{cleveref}

\newtheorem{theorem}{Theorem}[section]
\newtheorem{lemma}[theorem]{Lemma}
\newtheorem{proposition}[theorem]{Proposition}
\newtheorem{corollary}[theorem]{Corollary}
\newtheorem{remark}[theorem]{Remark}
\newtheorem{claim}[theorem]{Claim}
\theoremstyle{definition}

\numberwithin{equation}{section}

\makeatletter
\addtotheorempostheadhook{%
  \Hy@MakeCurrentHrefAuto{statement}%
  \edef\osp@statementanchor{\@currentHref}%
  \AddToHookNext{para/begin}{%
    \Hy@raisedlink{%
      \hyper@anchorstart{\osp@statementanchor}\hyper@anchorend}}}
\makeatother

\newcommand{\R}{\mathbb{R}}
\newcommand{\Pb}{\mathbb{P}}
\newcommand{\eps}{\varepsilon}
\newcommand{\ip}[2]{\left\langle #1,#2\right\rangle}
\newcommand{\norm}[1]{\lVert #1\rVert_2}
\newcommand{\OS}{\operatorname{OS}}
\newcommand{\Res}{\mathrm{res}}
\newcommand{\tot}{\mathrm{tot}}
\newcommand{\err}{\mathrm{err}}
\newcommand{\anc}{\mathrm{anc}}
\newcommand{\1}{\mathbbm{1}}

\setlist[itemize]{leftmargin=1.4em,itemsep=2pt,topsep=3pt}
\setlist[enumerate]{leftmargin=1.7em,itemsep=2pt,topsep=3pt}
\title{The Online Steinitz Problem}
\author{
Matteo Russo\thanks{EPFL, Switzerland, \textsf{\{matteo.russo, jiaye.wei\}@epfl.ch}}
\and
Jiaye Wei\footnotemark[1]
} 
\date{}

\begin{document}
\maketitle

\begin{abstract}
\noindent Inspired by Steinitz's celebrated result (1913) and recent
online vector-balancing results of Kulkarni et al. (2024) and
Aden-Ali (2026), we introduce the \emph{online Steinitz problem}:
$T$ vectors in $\mathbb R^d$, each of $\ell_2$-norm at most $1$ and
with total sum zero, arrive one at a time and must be assigned
immediately and irrevocably to unused positions of an array.
The goal is to keep every prefix sum of the resulting order small.
We first obtain deterministic prefix-discrepancy bounds of
$O(\log T)$ for $d=1$ and $O(\sqrt T)$ for $d\ge2$.
Our main contribution is a general reduction that converts
deterministic and randomized online vector-balancing algorithms,
which assign signs to arriving vectors, into online ordering algorithms.
Combined with recent vector-balancing results, this yields a
randomized algorithm with prefix discrepancy
$O(\log T+\sqrt{d\log T})$ and running time $O(dT\log T)$.
The guarantee holds with high probability against an oblivious
adversary. The reduction also gives bounds for sparse vectors and
extends to inputs in arbitrary compact convex bodies.
Finally, we prove matching lower bounds in both the deterministic
and randomized cases.
\end{abstract}

\input{sections/01_introduction}
\clubpenalty=10000
\widowpenalty=10000
\input{sections/02_reduction}
\input{sections/03_deterministic}

\input{sections/04_corollaries}
\input{sections/05_lower_bounds}

\section*{Acknowledgments}

The authors thank Rom Pinchasi for the fruitful discussions and initial ideas for the deterministic algorithm. The authors used ChatGPT to assist with mathematical arguments. Specifically, the initial quota correction argument in the reduction was suggested by ChatGPT-5.5 and subsequently refined by the authors. The lower-bound argument was suggested by ChatGPT-6 Astra. The authors take full responsibility for all the content of this work.

\bibliographystyle{alpha}
\bibliography{references}

\end{document}

%% file: sections/01_introduction.tex
\section{Introduction}
\label{sec:introduction}

The classical Steinitz lemma asks how to order a zero-sum family of bounded vectors so that every prefix sum remains small~\cite{Steinitz13}. For vectors $v_1,\ldots,v_T\in\R^d$ with
$\|v_t\|_2\le1$ and $\sum_{t=1}^T v_t=0$, the goal is to prove there exists an ordering distributes cancellation throughout the sequence instead
of postponing it until the end. Indeed, Steinitz \cite{Steinitz13} proved that there is a permutation
$\pi:[T]\to[T]$ such that
\[
 \max_{0\le k\le T}\left\|\sum_{j=1}^k v_{\pi(j)}\right\|_2\le d.
\]
Establishing an $O(\sqrt d)$ upper bound for the Euclidean Steinitz
constant has been a long-standing open problem. On 6 October 2026,
OpenAI announced a proof of this bound
\cite{OpenAIannouncement26}, which, if confirmed, would determine
the constant depending on $d$ up to absolute factors.\footnote{The authors have not independently verified the announced proof, which none of this paper's results requires.} Beyond the importance of the mathematical question itself, the principle of the Steinitz lemma serves as a standard tool in discrepancy theory and integer programming, e.g.,
\cite{Chobanyan1994,GrinbergSevastyanov1980,EisenbrandWeismantel2020}.

We introduce the \emph{online Steinitz problem} (OSP). In this problem, there are $T$ initially
empty output positions, and vectors $v_1,\ldots,v_T$ that arrive sequentially with $\|v_t\|_2\le1$ and $\sum_{t=1}^T v_t=0$, identically to the offline case.
Upon seeing $v_t$, the algorithm must assign it immediately and irrevocably
to an unused position $p_t\in[T]$, thereby building a permutation $\pi$ on the fly such that $\pi(p_t)=t$. The completed output is therefore $v_{\pi(j)}$ and its \emph{prefix discrepancy} by the end of the assignment is $\max_{0\le k\le T}\|\sum_{j=1}^k v_{\pi(j)}\|_2$. Although the prefix discrepancy is an identical measure in both offline and the online case, the requirement of an immediate and irrevocable placement and the limited future knowledge leads to a substantially harder algorithmic task that we seek to investigate:\medskip

\noindent \textbf{Question.} \emph{What is the worst-case online prefix discrepancy for immediate and irrevocable vector decisions, and how does this compare with the offline optimum?}\medskip

We study this question under the \emph{oblivious adversarial} model, which fixes the entire input independently of the algorithm's randomness, and also under the \emph{adaptive adversarial} model, which may choose each vector after observing the previous assignments. This means that, against the former kind of adversary, randomness may help, while against the latter it does not. We therefore write
$\OS_{\rm rand}(d,T,\delta)$ for the smallest discrepancy bound that a
randomized algorithm can guarantee with probability at least $1-\delta$
on every fixed input. The corresponding optimal deterministic
worst-case bound against adaptive inputs is denoted by
$\OS_{\det}(d,T)$.

The problem is reminiscent of \emph{online vector balancing} (OVB) \cite{BansalJ0S20,AlweissLS21,KulkarniReisRothvoss2024,AdenAli2026}: in that case, an OVB algorithm chooses a sign for each arriving vector and controls signed prefix sums in arrival order. On the other hand, an OSP algorithm chooses among all vacant positions, and one placement may affect many output prefixes, which are evaluated only after the stream ends, and not after each arrival. We defer the reader to the survey of Bansal and Nikolov~\cite{BansalN26} for an extensive treatment of algorithmic discrepancy in its various forms, including offline and online.

\subsection{Our results}
\label{sec:results}

Despite their differences, our main upper-bound contribution is a general reduction from OSP to OVB. We first describe this general reduction, followed by its deterministic and randomized applications,
further corollaries, and matching lower bounds.

\paragraph{A general reduction.}
In \cref{sec:blackbox}, we present a black-box reduction from online vector balancing to online Steinitz ordering. The reduction builds an OSP \emph{router} from OVB subroutines, which we call \emph{signers}. Its ordering guarantee is stated in terms of the signers’ final vector and count bounds, or alternatively their subgaussian estimates. As a result, the same routing framework can use a range of deterministic or randomized signers, and any improvement in their balancing performance immediately yields stronger ordering bounds. The reduction also applies when the total sum of the input vectors is arbitrary and unknown in advance, and it ensures that every prefix remains close to its proportional share of this total. In addition, if the inputs lie in a known compact convex body $K$ of affine dimension $s$, then any Euclidean proportional-prefix bound $\rho$ implies that the residuals lie in $(s\rho/2)(K-K)$, where, as usual, $K-K=\{x-y:x,y\in K\}$. This extension does not require $K$ to be symmetric, and it preserves deterministic and probabilistic guarantees under the same adversarial model.

\paragraph{Deterministic bounds.}
Using the OSP router given by the reduction instantiated with deterministic OVB signers, in \cref{sec:deterministic}, we show that the optimal deterministic discrepancy is 
\[
    \OS_{\rm det}(d,T) \in \begin{cases}
        O(\log T),&d=1,\\[2mm]
        O(\sqrt T),&d\ge2.
    \end{cases}
\]
As mentioned previously, both upper bounds hold against
adaptive adversaries, and the higher-dimensional lower bound already
holds in a fixed plane.

\paragraph{Tight bounds against an oblivious adversary.}
In \cref{sec:corollaries}, we instantiate the general reduction with Aden-Ali's randomized OVB signer \cite{AdenAli2026} and obtain
a prefix discrepancy of
\begin{equation}
    \OS_{\rm rand}(d,T,\delta)
     \in O\!\left(
        \log T+\sqrt{\bigl(d+\log(1/\delta)\bigr)\log T}
    \right),
    \label{eq:main-tight}
\end{equation}
with probability at least $1-\delta$. This guarantee is achieved efficiently, with the overall procedure having a running time of $O(dT\log T)$, and furthermore yields a guarantee of $ O\left(\log T + \sqrt{\log T \cdot \log(1/\delta)}\right)$ in the $\ell_\infty$-norm. Furthermore, we can instantiate our reduction with a different OVB signer for zero-sum inputs in $[-1,1]^d$ with at most $\kappa$ nonzero coordinates, and $\ell_2$-norm as large as $\sqrt\kappa$. Specifically, the signer of Altschuler--Tikhomirov
\cite{AltschulerTikhomirov2026} gives
$O(\sqrt\kappa\log T)$ coordinatewise discrepancy whenever the
failure bound in \cref{cor:at-overview} is at most $\delta$.

\paragraph{Matching lower bounds.}
In \cref{sec:lower}, all hard inputs consist of an initial stream
followed by identical vectors that cancel its total.
An all-ones initial stream and Tijdeman--Wagner~\cite{TW80} give
an $\Omega(\log T)$ lower bound; coordinate cycles and Roth's
method~\cite{Roth54} give
$\Omega(\sqrt{d_*\log(2+T/d_*)})$, where $d_*=\min\{d,T\}$.
Both hold for every $0<\delta\le1/2$.
Their sum matches the upper bound of \eqref{eq:stopped-tree} 
whenever $T^{-\beta}\le\delta\le1/2$ for fixed $\beta>0$.
An adaptive initial stream in a fixed plane gives the matching
deterministic $\Omega(\sqrt T)$ lower bound for $d\ge2$.

\subsection{Related work}
\label{sec:related-work}

\paragraph{The Steinitz lemma and its extensions.}
In the offline setting, the classical Steinitz lemma \cite{Steinitz13} guarantees that a
zero-sum family of bounded vectors admits an ordering whose prefix
sums are bounded independently of the number of vectors.
Grinberg and Sevastyanov \cite{GrinbergSevastyanov1980} established
the bound $d$ in every $d$-dimensional normed space.
For the $\ell_2$-norm, improving the dimension dependence to
$O(\sqrt d)$ remains a longstanding conjecture, with a matching
$\Omega(\sqrt d)$ lower bound.
Recent work of Ambrus and Heck \cite{AmbrusHeck2026} relates this
problem to ordering nearly unit vectors with arbitrary total sum,
measuring each prefix against its proportional share of that total.
Another direction concerns simultaneous rearrangements:
Oertel et al. \cite{OertelPaatWeismantel2024}
proved a colorful Steinitz lemma that permutes several sequences
while controlling their combined prefixes independently of the
number of sequences. B\'ar\'any \cite{Barany2024Matrix} improved
their $40d^5$ bound to $4d-2$.
Steinitz lemma also has algorithmic applications beyond discrepancy, as for instance, the work of
Eisenbrand and Weismantel \cite{EisenbrandWeismantel2020} which used it
to establish proximity bounds and faster algorithms for integer
programming.

\paragraph{Online vector balancing and scheduling.}
Online vector balancing replaces rearrangement by sequential sign
choices. Bansal et al. \cite{BansalJ0S20}
studied stochastic inputs drawn independently from an arbitrary
distribution and obtained applications to geometric discrepancy.
For oblivious adversarial inputs, Alweiss et al.
\cite{AlweissLS21} introduced the self-balancing walk, giving
efficient algorithms with logarithmic $\ell_\infty$ prefix
discrepancy for Euclidean unit-ball vectors.
Kulkarni et al. \cite{KulkarniReisRothvoss2024}
subsequently obtained a universal subgaussian guarantee for each
signed prefix. This implies $O(\sqrt{\log T})$ maximum
$\ell_\infty$ prefix discrepancy with high probability, matching
their lower bound against oblivious adversaries.
Aden-Ali \cite{AdenAli2026} obtained the same asymptotic guarantee
with an $O(dT)$-time algorithm.
For coordinate-bounded vectors with at most $\kappa$ nonzero
coordinates, Altschuler and Tikhomirov
\cite{AltschulerTikhomirov2026} obtained
$O(\sqrt{\kappa})$ online $\ell_\infty$ prefix discrepancy above
a near-logarithmic sparsity threshold. Online vector scheduling \cite{ImEtAl2019} likewise requires
irrevocable allocation of arriving vectors, but distributes
nonnegative demand vectors among machines to optimize their loads,
whereas OSP assigns vectors to ordered positions and exploits
cancellation to control all prefix sums.

\paragraph{Signing, rearrangement, and geometric discrepancy.}
The connection between signing and rearrangement predates their
online versions. B\'ar\'any and Grinberg \cite{BaranyGrinberg1981}
proved that every prescribed sequence of unit vectors in a
$d$-dimensional normed space admits offline signs for which all
signed prefixes have norm at most $2d-1$.
Chobanyan's transference theorem \cite{Chobanyan1994} shows that
a uniform bound for this offline prefix-signing problem also
bounds the Steinitz constant, without increasing the bound. However, this does not transfer to the online setting as argued in \Cref{sec:reduction-remarks}.
Controlling only the final signed sum is a different problem.
For example, Ambrus and Bozzai \cite{AmbrusBozzai2024} proved
sharp Euclidean bounds for colorful vector balancing, where one
selects a vector from each prescribed family and controls the
sum of the selected vectors.
Our lower bounds use another connection to discrepancy: points record
arrival times and assigned positions. Schmidt's planar rectangle-discrepancy
theorem~\cite{Schmidt1972} gives the deterministic scalar bound;
Tijdeman--Wagner's strengthening~\cite{TW80} and a restricted form of
Roth's orthogonal-function argument~\cite{Roth54} yield the oblivious bounds.
For a broader account of these topics and their algorithmic
connections, see the monograph manuscript of Bansal and Nikolov
\cite{BansalN26}.

\paragraph{Online sorting and online TSP.}
Our placement model is also closely related to online sorting
\cite{AamandEtAl2023} and its extension to online TSP
\cite{AbrahamsenEtAl2024, Bertram25}: arriving numbers or points must be
assigned immediately and irrevocably to unused array positions.
It is worth noting that the work of Azar et al. \cite{AzarPV26}
further studies the tradeoff between array size and competitiveness
for online metric TSP, obtaining a deterministic
$O(\log^3 T)$-competitive algorithm with $2T$ cells.
Their placement constraint coincides with ours when the array has
exactly $T$ cells, but their objective controls distances between
consecutive entries rather than cumulative vector sums.
Their tree-based placement scheme organizes points by metric
proximity, whereas our reduction uses online signing to balance
the sums routed to subintervals. More generally, these problems minimize the total distance between consecutive
entries, whereas OSP controls the norm of every prefix sum.
Geometrically, their inputs are points to connect by a short path, while
ours are steps to arrange into a walk that stays close to the
origin through cancellation.

%% file: sections/02_reduction.tex
\section{The signing-to-ordering reduction}
\label{sec:overview}
\label{sec:blackbox}

The reduction repeatedly divides the output positions into two nearly
equal intervals. At each division, an online signer proposes a child
for the arriving vector, and the router can only override this choice when it has to avoid
an interval that is already full. We first describe this algorithm,
then state the signing guarantees that suffice and prove how its local
splits control every output prefix.

\subsection{The ordering algorithm}
\label{sec:model}

We first introduce some notation, beginning with the definition of residuals.

\paragraph{Residuals.}
For the completed ordering $v_\pi=(v_{\pi(1)},\ldots,v_{\pi(T)})$,
write $\tot:=\sum_{t=1}^T v_t$ and define the prefix residual for each $k \in [T]$ and the maximum residual by
\begin{equation}
    \Res_k(v_\pi):=
    \sum_{j=1}^k v_{\pi(j)}-\frac{k}{T}\tot,
    \qquad
    \Res_{\max}(v_\pi):=
    \max_{0\le k\le T}\norm{\Res_k(v_\pi)}.
    \label{eq:residual}
\end{equation}
These residuals measure deviations from proportional shares of the total and become ordinary prefix sums when $\tot=0$. Hence, for some upper bound $U \ge 0$ achieved deterministically or in high probability, we have that, if $\Res_{\max} \le U$ then $\OS \le U$, since $\tot=0$.
We next describe the position tree and the meta-algorithm needed for our reduction.

\paragraph{The position tree.}
Our algorithm assigns each arriving vector to an output position by
routing it through a balanced binary tree.
This tree has the output positions as its leaves, and so has height $h = \lceil\log_2 T\rceil$.
Each node $u$ represents a consecutive interval $I_u\subseteq[T]$ which consists of the output positions labeling the leaves in the subtree rooted at $u$.
The node $u$ has capacity $m_u:=|I_u|$ with children $u_L,u_R$ of capacities $\lceil m_u/2\rceil$ and $\lfloor m_u/2\rfloor$, respectively.

As illustrated in \cref{alg:tree-routing,fig:intro-routing}, each
internal node uses a persistent local signer $\mathsf{Sign}_u$ to
propose $+1$ (left) or $-1$ (right). The router follows this proposal
unless the requested child is full, in which case it uses the other
child. Repeating the decision until a leaf is reached assigns the
vector's output position immediately. During the descent, the
algorithm records in $N_c$ how many positions are already filled in
child $c$, and in $X_c$ the sum of the corresponding vectors.
The arguments of $\mathsf{Sign}_u$ are the arriving vector $v_t$, the child sums $X_{u_L},X_{u_R}$, and the remaining capacities $m_{u_L}-N_{u_L},m_{u_R}-N_{u_R}$. Here the semicolon separates $v_t$
from the routing statistics.
In particular, the scalar rule in \cref{sec:deterministic} uses these
statistics directly, while, for a Euclidean signing subroutine,
the node submits the augmented vector $(v_t,1)/\sqrt2$, so the extra
coordinate balances left and right counts along with the original
vector coordinates.

Each local signer receives every vector reaching its node, even after
one child fills, and updates its private state with its proposed sign.
We stress that only the router overrides the destination. Moreover, we also note that randomized copies have
independent tapes, independent of the input, and different rules may
be used at different node sizes.

\begin{algorithm}[tbp]
\caption{Tree routing with local proposal rules}
\label{alg:tree-routing}
\algrenewcommand{\algorithmicindent}{1em}
\algtext*{EndFor}
\algtext*{EndWhile}
\algtext*{EndIf}
\begin{algorithmic}[1]
\Require Horizon $T$, local routines $\mathsf{Sign}_u$, and vectors $v_1,\ldots,v_T$ arriving one at a time.
\Ensure An online permutation $\pi$, with $\pi(p_t)=t$.
\State Build the balanced position tree with capacities $m_u=|I_u|$.
\For{each internal node $u$}
    \State Initialize $\mathsf{Sign}_u$ once.
    \For{$c\in\{u_L,u_R\}$}
        \State $N_c\gets0$, \quad $X_c\gets0$.
    \EndFor
\EndFor
\For{$t=1,\ldots,T$}
    \State Receive $v_t$ and set $u\gets\mathrm{root}$.
    \While{$u$ is not a leaf}
        \State $\eps\gets\mathsf{Sign}_u
        \bigl(v_t;X_{u_L},X_{u_R},
        m_{u_L}-N_{u_L},m_{u_R}-N_{u_R}\bigr)$.
        \State $\widehat\eps\gets\eps$.
        \Comment{Proposed versus actual routing sign}
        \If{$N_{u_L}=m_{u_L}$}
            \State $\widehat\eps\gets-1$.
        \ElsIf{$N_{u_R}=m_{u_R}$}
            \State $\widehat\eps\gets+1$.
        \EndIf
        \If{$\widehat\eps=+1$}
            \State $c\gets u_L$.
        \Else
            \State $c\gets u_R$.
        \EndIf
        \State $N_c\gets N_c+1$, \quad $X_c\gets X_c+v_t$.
        \State $u\gets c$.
    \EndWhile
    \State $p_t\gets\min I_u$, \quad $\pi(p_t)\gets t$.
\EndFor
\end{algorithmic}
\end{algorithm}

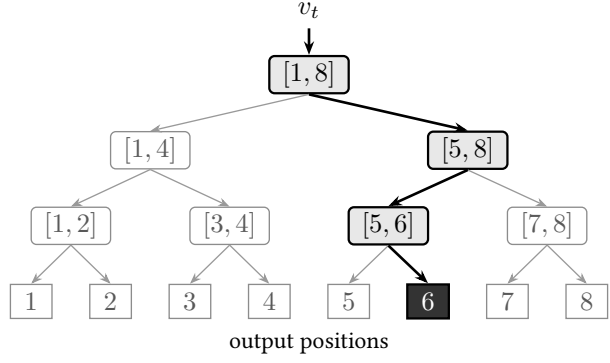
\begin{figure}[tbp]
\centering
\begin{tikzpicture}[
    x=1.05cm,y=1.00cm,
    every node/.style={font=\small},
    interval/.style={
        draw=black!45, text=black!65, rounded corners=2pt,
        minimum width=1.05cm, minimum height=0.48cm,
        inner sep=2pt, line width=0.5pt
    },
    slot/.style={
        draw=black!45, text=black!65,
        minimum width=0.55cm, minimum height=0.45cm,
        inner sep=0pt, line width=0.5pt
    },
    active/.style={draw=black, text=black, fill=black!9,
                   line width=0.85pt},
    edge/.style={-{Stealth[length=1.5mm]},
                 draw=black!40, line width=0.5pt},
    route/.style={-{Stealth[length=1.8mm]},
                  draw=black, line width=1pt}
]
\node[interval,active] (root) at (4.5,3) {$[1,8]$};
\node[interval]        (L)    at (2.5,2) {$[1,4]$};
\node[interval,active] (R)    at (6.5,2) {$[5,8]$};
\node[interval]        (LL)   at (1.5,1) {$[1,2]$};
\node[interval]        (LR)   at (3.5,1) {$[3,4]$};
\node[interval,active] (RL)   at (5.5,1) {$[5,6]$};
\node[interval]        (RR)   at (7.5,1) {$[7,8]$};

\foreach \i in {1,...,8} {
    \ifnum\i=6
        \node[slot,draw=black,fill=black!80,text=white,
              line width=0.85pt] (s\i) at (\i,0) {$\i$};
    \else
        \node[slot] (s\i) at (\i,0) {$\i$};
    \fi
}

\foreach \parent/\child in {
    root/L, L/LL, L/LR, R/RR,
    LL/s1, LL/s2, LR/s3, LR/s4, RL/s5, RR/s7, RR/s8}
    \draw[edge] (\parent.south) -- (\child.north);

\node (incoming) at (4.5,3.85) {$v_t$};
\foreach \parent/\child in {incoming/root, root/R, R/RL, RL/s6}
    \draw[route] (\parent.south) -- (\child.north);

\node[font=\footnotesize] at (4.5,-0.55) {output positions};
\end{tikzpicture}
\caption{Routing an arriving vector $v_t$ when $T=8$.
At each visited internal node, the local rule selects one of
the two child intervals. The highlighted path assigns $v_t$
to output position $6$.}
\label{fig:intro-routing}
\end{figure}

\subsection{The reduction guarantee}

To state the guarantee of \cref{alg:tree-routing}, we first identify
which local splits contribute to the discrepancy of a given output
prefix. Since the node intervals are consecutive, a prefix cuts
through at most one interval at each depth. For each prefix $[k]$, where $0\le k\le T$, the corresponding \emph{boundary path} $P(k)$ consists of the internal nodes whose
intervals meet both $[k]$ and $[T]\setminus[k]$, with $|P(k)|\le h$ and
$P(0)=P(T)=\varnothing$. Every ancestor of a node in $P(k)$ also belongs to $P(k)$, hence these nodes form a path which is ordered by increasing depth. 
We will show that the prefix residual $\Res_k(v_\pi)$ is a
weighted sum of the split errors along this path. The reduction
therefore depends on how well the local signers control these errors,
which we quantify next.

\paragraph{Signers.}
An online signer returns $\eps_i\in\{-1,+1\}$ upon receiving $z_i$.
We say it has subgaussian scale $\Sigma\in[0,\infty)$ if, for every
fixed stream $z_1,\ldots,z_m$ in the Euclidean unit ball $B_2^D$
and every $\theta\in\R^D$,
\begin{equation}
    \mathbb E\exp\!\left(
        \ip{\theta}{\sum_{i=1}^m\eps_i z_i}
    \right)
    \le
    \exp\!\left(\frac{\Sigma^2}{2}\norm{\theta}^2\right).
    \label{eq:terminal-mgf}
\end{equation}
The expectation is over its randomness, independent of the input, which gives Gaussian-type tails in every direction.
Our randomized reduction uses a common scale $\Sigma$ in dimension
$d+1$ for all node capacities, and needs only this final signed-sum
guarantee, rather than a guarantee for every signed prefix.

The following theorem separates the property needed from the signer
from the resulting ordering guarantee. Terminal vector and count
bounds give a sum of local errors, while a subgaussian bound permits a
sharper estimate for their combined effect.

\begin{theorem}[Signing-to-ordering reduction]
\label{thm:blackbox-main}
\Cref{alg:tree-routing} assigns every vector immediately to an
unused position and does not need to know the input total. Randomized
local signers use independent tapes, independent of the input.
\begin{enumerate}[label=\textup{(\roman*)}]
\item \textbf{Terminal vector and count bounds.}
Suppose the input vectors have norm at most one in a fixed norm
$\lVert\cdot\rVert$. At a node receiving $x_1,\ldots,x_m$, suppose the
signer's proposed signs satisfy
\[
 \left\lVert\sum_{i=1}^m\eps_i x_i\right\rVert\le V(m),
 \qquad
 \left|\sum_{i=1}^m\eps_i\right|\le W(m).
\]
Whenever these inequalities hold at all internal nodes, every output
prefix satisfies
\begin{equation}
 \lVert\Res_k(v_\pi)\rVert
 \le\sum_{u\in P(k)}
 \left(\frac{V(m_u)+W(m_u)}2+\1\{m_u\text{ odd}\}\right).
 \label{eq:reduction-norm}
\end{equation}
If they hold for every local stream, this is a deterministic guarantee
against adaptive inputs. If, for each fixed length-$m$ stream, they
hold jointly with probability at least $1-p(m)$, then for every
oblivious input \eqref{eq:reduction-norm} holds simultaneously for all
prefixes with probability at least $1-\sum_u p(m_u)$.

\item \textbf{Subgaussian terminal sums.}
Suppose the inputs are in the Euclidean unit ball and the signer has
scale $\Sigma\ge1$ in the sense of \eqref{eq:terminal-mgf}, in dimension
$d+1$ and for every node capacity. Then, for every oblivious input
and $0<\delta<1$, with probability at least $1-\delta$,
\begin{equation}
 \Res_{\max}(v_\pi)
   \le C\Sigma\left(\log T+\sqrt{\log T\bigl(d+\log(1/\delta)\bigr)}\right),
 \label{eq:reduction-subgaussian}
\end{equation}
where $C$ is an absolute constant.
\end{enumerate}
Both parts require only the \emph{terminal} guarantee for each local
stream. For zero-sum inputs, the residuals are the ordinary prefix sums.
\end{theorem}

Part~\textup{(i)} will give the deterministic Euclidean and sparse-vector
bounds. Part~\textup{(ii)} will give the randomized Euclidean bound without
paying a separate norm bound at every node. The sharper dimension-dependent
bound follows by using deterministic routing near the leaves.

\subsection{From local splits to output prefixes}

We prove \cref{thm:blackbox-main} by first relating the output prefixes
to the errors of the local splits. The identity in \cref{lem:one-path}
shows that each prefix depends only on its boundary path;
\cref{lem:quota-repair} then bounds each split error in terms of the
signer's proposed signs and the corrections needed to enforce
capacities. Combining these two facts with the triangle inequality
will prove part~\textup{(i)}. For part~\textup{(ii)}, we will use the
conditional subgaussian estimates to bound the sum along the path.

Before continuing, we define quantities only needed in the analysis. Namely, we let $\anc(u)$ be the set of strict ancestors of $u$ in the tree and, throughout, conditioning on $\anc(u)$ means conditioning on the
complete random tapes of all nodes in $\anc(u)$. Furthermore, for the completed ordering,
\begin{equation}
    \tot_u:=\sum_{j\in I_u}v_{\pi(j)},
    \qquad
    \err_u:=\tot_{u_L}-\frac{m_{u_L}}{m_u}\tot_u.
    \label{eq:local-error}
\end{equation}
Thus $\err_u$ is the left child's error relative to its proportional
share of the parent total, and the right child's error is $-\err_u$. To track the same discrepancy within an interval,
for $0\le s\le m_u$ we set
\[
    F_u(s):=\sum_{j=\min I_u}^{\min I_u+s-1}v_{\pi(j)}
    -\frac{s}{m_u}\tot_u.
\]
This is the deviation of the first $s$ positions in $I_u$ from their
proportional share of $\tot_u$, so $F_{\rm root}(k)=\Res_k(v_\pi)$ and
$F_u(0)=F_u(m_u)=0$. For $u\in P(k)$, we also set
\[
    \alpha_{u,k}:=\min\left\{
       \frac{|I_u\cap[k]|}{m_{u_L}},\,
       \frac{|I_u\setminus[k]|}{m_{u_R}}
    \right\}\in[0,1].
\]
When the prefix ends in the left child, this coefficient is the
fraction of that child included in the prefix; when it ends in the
right child, it is the fraction of that child excluded from the
prefix. At the boundary between the children, it equals one. The
following identity shows that these fractions are exactly the weights
with which the local split errors contribute to the prefix residual.

\begin{lemma}[one-boundary-path identity]
\label{lem:one-path}
For every output prefix $k$, the boundary path and coefficients above
satisfy
\begin{equation}
 \Res_k(v_\pi)=\sum_{u\in P(k)}\alpha_{u,k}\err_u.
 \label{eq:one-path}
\end{equation}
\end{lemma}

\begin{proof}
The cases $k=0,T$ are immediate. At a node $u\in P(k)$, write
$s=|I_u\cap[k]|$. If $s\le m_{u_L}$, the local prefix lies entirely in
the left child, and adding and subtracting
$(s/m_{u_L})\tot_{u_L}$ gives
\begin{align*}
    F_u(s)
    =F_{u_L}(s)+\frac{s}{m_{u_L}}\tot_{u_L}
      -\frac{s}{m_u}\tot_u =F_{u_L}(s)+\frac{s}{m_{u_L}}\err_u.
\end{align*}
Here $\alpha_{u,k}=s/m_{u_L}$, since
$(m_u-s)/m_{u_R}\ge1$. If instead $s=m_{u_L}+t$ with
$0\le t\le m_{u_R}$, the prefix contains the whole left child and the
first $t$ positions of the right child, so
\begin{align*}
    F_u(m_{u_L}+t)
    &=F_{u_R}(t)+\tot_{u_L}+\frac{t}{m_{u_R}}\tot_{u_R}
      -\frac{m_{u_L}+t}{m_u}\tot_u\\
    &=F_{u_R}(t)+\err_u
      +\frac{t}{m_{u_R}}
       \left(\tot_{u_R}-\frac{m_{u_R}}{m_u}\tot_u\right) =F_{u_R}(t)+\left(1-\frac{t}{m_{u_R}}\right)\err_u.
\end{align*}
In this case $\alpha_{u,k}=1-t/m_{u_R}$, since
$(m_{u_L}+t)/m_{u_L}\ge1$.
Starting with $F_{\rm root}(k)=\Res_k(v_\pi)$, each substitution
therefore adds the term $\alpha_{u,k}\err_u$ and leaves only the
residual in the child crossed by the prefix boundary. More explicitly,
when the recursion reaches a node $u$, it has established
\[
 \Res_k(v_\pi)
 =\sum_{\substack{v\in P(k)\\v\text{ a strict ancestor of }u}}
    \alpha_{v,k}\err_v
   +F_u\bigl(|I_u\cap[k]|\bigr).
\]
The formulas above preserve this identity as we descend to the unique
child that meets both the prefix and its complement. If the prefix
ends between the two children, then $s=m_{u_L}$ and the remaining
residual is $F_{u_L}(m_{u_L})=0$; a residual at a leaf is also zero.
Thus the recursion visits exactly the nodes of $P(k)$ and leaves no
residual term, proving \eqref{eq:one-path}.
\end{proof}

The path identity reduces the ordering problem to controlling
$\err_u$ at each node. The signer supplies a proposed split, whereas
$\err_u$ comes from the split after capacities have been enforced.
The next lemma connects the two: it expresses $\err_u$ as a centered
signed sum plus a correction controlled by the signer's count
imbalance. It also records the conditional estimates needed later to
combine these errors along a path.

\begin{lemma}[one-node quota repair]
\label{lem:quota-repair}
Let $u$ be an internal node of capacity $m$, receiving Euclidean
unit-ball vectors $x_1,\ldots,x_m$, and write $\mu=m^{-1}\sum_i x_i$.
Let $\eps_i$ be the proposed signs and $\hat \varepsilon_i$ the quota-corrected
routing signs, where $+1$ means left, and put $q=\sum_i\eps_i, Y=\sum_i\eps_i(x_i-\mu)$ for succinctness. Then the local proportional error of \eqref{eq:local-error} has a
decomposition
\[
    \err_u=\frac12Y+R,
    \qquad \norm{R}\le |q|+1.
\]
Under the subgaussian hypothesis of \cref{thm:blackbox-main},
for a fixed input, conditionally on the routing tapes at the strict
ancestors of $u$, for every unit vector $z\in\R^d$ and every
$\lambda\in\R$,
\begin{align}
 \mathbb E\!\left[e^{\lambda\ip{z}{Y}/2}\mid\text{strict ancestors}\right]
 &\le e^{\Sigma^2\lambda^2/2},
 \label{eq:y-mgf}\\
 \mathbb E\!\left[e^{\lambda q}\mid\text{strict ancestors}\right]
 &\le e^{\Sigma^2\lambda^2}.
 \label{eq:q-mgf}
\end{align}
\end{lemma}

\begin{proof}
We first count the changes made by quota repair. Since the children
receive exactly their prescribed numbers of vectors, the corrected
signs satisfy $\sum_i\hat \varepsilon_i=\1\{m\text{ odd}\}$. Once a proposal is overridden, one
child is full, so every later override changes a proposal in the same
direction. Each change alters the signed count by two, and comparing
the proposed count $q$ with the corrected count $\1\{m\text{ odd}\}$ therefore gives
\begin{equation}
 \text{number of overrides}=\frac12|q-\1\{m\text{ odd}\}|.
 \label{eq:override-count}
\end{equation}
\cref{fig:deterministic-quota-repair} illustrates this identity for even $m$ and $q=2$.
To translate this into a bound on the split error, we simply use 
$m_{u_L}/m=(1+\1\{m\text{ odd}\}/m)/2$ and $\tot_u=m\mu$ to write
\[
 \err_u
 =\frac12\left(\sum_i\hat \varepsilon_i x_i-\1\{m\text{ odd}\}\mu\right)
 =\frac12\sum_i\hat \varepsilon_i(x_i-\mu)
 =\frac12Y+\underbrace{\frac12\sum_i(\hat \varepsilon_i-\eps_i)(x_i-\mu)}_{\coloneqq R},
\]
Each overridden sign contributes one vector
$\pm(x_i-\mu)$, whose norm is at most two because both $x_i$ and
$\mu$ lie in the unit ball. Together with \eqref{eq:override-count},
this yields $\norm{R}\le |q-\1\{m\text{ odd}\}|\le |q|+1$.

We now turn to the conditional estimates, and fix the complete tapes of the strict
ancestors of $u$. The original input is fixed, and ancestor routing
never consults descendant signers, so these tapes determine the
entire stream reaching $u$, while its private tape remains independent
with its original distribution. We can therefore apply
\eqref{eq:terminal-mgf} to the augmented stream
$(x_i,1)/\sqrt2$, whose terminal signed sum is
$(\sum_i\eps_i x_i,q)/\sqrt2$.
Taking $\theta=(\lambda/\sqrt2)(z,-\ip{z}{\mu})$ gives
\[
 \left\langle
   \frac{\lambda}{\sqrt2}(z,-\ip{z}{\mu}),\,
   \frac1{\sqrt2}\left(\sum_i\eps_i x_i,q\right)
 \right\rangle
 =\frac{\lambda}{2}\ip{z}{Y},
 \qquad
 \norm{\theta}^2
 =\frac{\lambda^2}{2}\left(1+\ip{z}{\mu}^2\right)
 \le\lambda^2,
\]
which substituted in \eqref{eq:terminal-mgf} proves \eqref{eq:y-mgf}.
Similarly, we may choose $\theta=(0,\ldots,0,\sqrt2\lambda)$, which has squared norm
$2\lambda^2$ and inner product $\lambda q$ with the augmented signed
sum, giving \eqref{eq:q-mgf}.
\end{proof}

We thus turn to the proof of the terminal vector and count bounds in \Cref{thm:blackbox-main}.

\begin{figure}[tbp]
\centering
\begin{tikzpicture}[
    x=0.57cm,y=0.63cm,
    every node/.style={font=\scriptsize}
]
\node[anchor=east] at (-0.2,1.6) {arrival};
\node[anchor=east] at (-0.2,0.9) {proposed sign};
\node[anchor=east] at (-0.2,0) {routing sign};

\foreach \j/\sgn in {1/+,2/+,3/-,4/+,5/-,6/+,7/-,8/+}{
    \node at (\j-0.5,1.6) {$\j$};
    \node at (\j-0.5,0.9) {$\sgn$};
}

\fill[black!12] (7,-0.32) rectangle (8,0.32);
\foreach \j/\sgn in {1/+,2/+,3/-,4/+,5/-,6/+,7/-,8/-}
    \node at (\j-0.5,0) {$\sgn$};

\draw[->] (7.5,0.63)--(7.5,0.35);
\draw[dashed] (6,1.22)--(6,-0.37);

\node[anchor=west] at (8.3,0.9) {five left, three right};
\node[anchor=west] at (8.3,0) {four left, four right};
\node[align=center] at (4,-1.0)
    {Each child has capacity four.\\
     The left child fills after arrival $6$;
     proposal $8$ is overridden.};
\end{tikzpicture}
\caption{Enforcing capacities without changing the signer's
internal decisions. The proposed signed count is $2$, the required
count is $0$, and exactly one sign is overridden.}
\label{fig:deterministic-quota-repair}
\end{figure}

\begin{proof}[Proof of \cref{thm:blackbox-main} (i)]
First, note that the same quota calculation applies in any norm. Moreover, with the notation of
\cref{lem:quota-repair}, each override changes the signed vector sum
by a vector of norm at most two. Hence \eqref{eq:override-count} and
the two terminal bounds assumed in part~\textup{(i)} of the statement give
\[
 \left\lVert\sum_i\hat \varepsilon_i x_i\right\rVert
 \le \left\lVert\sum_i\eps_i x_i\right\rVert+|q-\1\{m\text{ odd}\}|
 \le V(m)+W(m)+\1\{m\text{ odd}\}.
\]
We now have
\begin{equation}
 \lVert\err_u\rVert = \left\lVert \frac12 \left(\sum_i\hat \varepsilon_i x_i-\frac{\1\{m\text{ odd}\}}{m}\tot_u\right)\right\rVert
 \le\frac{V(m)+W(m)}2+\1\{m\text{ odd}\},
 \label{eq:norm-tail-local}
\end{equation}
where the equality holds by definition of the exact proportional error $\err_u$ and the inequality by the fact that $\lVert\tot_u/m\rVert\le1$. Applying the triangle inequality to \cref{lem:one-path} now gives
\[
 \lVert\Res_k(v_\pi)\rVert
 \le\sum_{u\in P(k)}\alpha_{u,k}\lVert\err_u\rVert
 \le\sum_{u\in P(k)}
 \left(\frac{V(m_u)+W(m_u)}2+\1\{m_u\text{ odd}\}\right),
\]
which is \eqref{eq:reduction-norm}. This argument holds for every
realized input and collection of proposed signs, so a deterministic
signer satisfying the terminal bounds on every stream also gives the
claimed guarantee against adaptive inputs.

For the probabilistic statement, fix an oblivious input and say that
a node fails if either of its two terminal bounds is violated.
Conditioning on its strict-ancestor tapes fixes its entire input
stream while leaving its own tape independent, just as in
\cref{lem:quota-repair}. The assumed fixed-stream guarantee thus bounds
its conditional failure probability by $p(m_u)$. Taking expectations
and then a union bound gives
\begin{align*}
 \Pb\!\left(\exists \text{ node }u\text{ fails}\right)
 &\le\sum_{u\text{ internal}}\Pb\!\left(\text{node }u\text{ fails}\right) =\sum_{u\text{ internal}}\mathbb E\!\left[
   \Pb\!\left(\text{node }u\text{ fails}
     \,\middle|\,\anc(u)\right)\right] \le\sum_{u\text{ internal}}p(m_u).
\end{align*}
Outside this failure event, all local bounds hold and the pathwise
argument above applies simultaneously to every prefix, which gives the first part of the theorem.
\end{proof}

\subsection{Combining randomized errors}

For part~\textup{(ii)} of \cref{thm:blackbox-main}, we use the
decomposition in \cref{lem:quota-repair} to combine the centered
signed sums before taking their norm, and then bound the quota
corrections through the count imbalances. The errors along a path
need not be independent, because ancestor decisions determine the
streams reaching their descendants. The conditional estimates in
\eqref{eq:y-mgf} and \eqref{eq:q-mgf} let us handle this dependence by conditioning
successively from the deepest node to the root.

\begin{lemma}[conditional composition]
\label{lem:path-composition}
Assume the subgaussian condition of \cref{thm:blackbox-main}, and
write $Y_u,q_u,R_u$ for the quantities in \cref{lem:quota-repair} at
node $u$. There is an absolute constant $C$ such that, for every
prefix $k$ and every $s>0$,
\begin{align}
 \Pb\!\left(
 \left\|\frac12\sum_{u\in P(k)}\alpha_{u,k}Y_u\right\|_2
 >C\Sigma\sqrt{h(d+s)}\right)
 &\le e^{-s},
 \label{eq:path-vector-tail}\\
 \Pb\!\left(
 \sum_{u\in P(k)}\alpha_{u,k}|q_u|
 >C\Sigma\sqrt{h(h+s)}\right)
 &\le e^{-s}.
 \label{eq:path-count-tail}
\end{align}
\end{lemma}

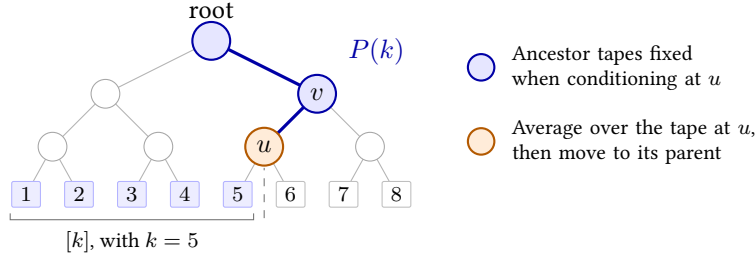
\begin{figure}[tbp]
\centering
\begin{tikzpicture}[
    x=0.70cm,y=0.70cm,
    every node/.style={font=\small},
    plain/.style={circle,draw=black!35,fill=white,
                  minimum size=3.8mm,inner sep=0pt},
    ancestor/.style={circle,draw=blue!65!black,fill=blue!10,
                     line width=0.8pt,minimum size=5mm,inner sep=1pt},
    deepest/.style={circle,draw=orange!70!black,fill=orange!15,
                    line width=0.8pt,minimum size=5mm,inner sep=1pt}
]
\coordinate (root) at (4,3.0);
\coordinate (left) at (2,2.0);
\coordinate (right) at (6,2.0);
\coordinate (ll) at (1,1.0);
\coordinate (lr) at (3,1.0);
\coordinate (rl) at (5,1.0);
\coordinate (rr) at (7,1.0);
\foreach \a/\b in {root/left,root/right,left/ll,left/lr,right/rl,right/rr}
    \draw[black!30] (\a)--(\b);
\foreach \a/\b/\c in {ll/0.5/1.5,lr/2.5/3.5,rl/4.5/5.5,rr/6.5/7.5}{
    \draw[black!30] (\a)--(\b,0.1);
    \draw[black!30] (\a)--(\c,0.1);
}
\draw[blue!65!black,line width=1.2pt] (root)--(right)--(rl);
\foreach \a in {left,ll,lr,rr}
    \node[plain] at (\a) {};
\node[ancestor] at (root) {};
\node[anchor=south] at (4,3.23) {root};
\node[ancestor] at (right) {$v$};
\node[deepest] at (rl) {$u$};
\node[blue!65!black,anchor=west] at (6.4,2.8) {$P(k)$};
\foreach \j in {6,7,8}
    \node[draw=black!25,fill=white,rounded corners=0.7pt,
          minimum width=3.8mm,minimum height=3.4mm,inner sep=0pt,
          font=\scriptsize] at (\j-0.5,0.1) {$\j$};
\foreach \j in {1,...,5}
    \node[draw=blue!30,fill=blue!7,rounded corners=0.7pt,
          minimum width=3.8mm,minimum height=3.4mm,inner sep=0pt,
          font=\scriptsize] at (\j-0.5,0.1) {$\j$};
\draw[dashed,black!55] (5,-0.35)--(5,0.6);
\draw[black!55] (0.2,-0.26)--(0.2,-0.38)--(4.8,-0.38)--(4.8,-0.26);
\node[anchor=north,font=\scriptsize] at (2.5,-0.4) {$[k]$, with $k=5$};
\node[ancestor,minimum size=3.5mm] at (9.1,2.5) {};
\node[anchor=west,align=left,font=\scriptsize,text width=3.6cm]
    at (9.5,2.5) {Ancestor tapes fixed\\when conditioning at $u$};
\node[deepest,minimum size=3.5mm] at (9.1,1.1) {};
\node[anchor=west,align=left,font=\scriptsize,text width=3.6cm]
    at (9.5,1.1) {Average over the tape at $u$,\\then move to its parent};
\end{tikzpicture}
\caption{Conditioning from the deepest node upward, illustrated for
$T=8$ and $k=5$. The highlighted path $P(k)$ contains the root,
$v$, and $u$. When conditioning on the strict-ancestor tapes of $u$,
the summands at the root and $v$ are fixed, so only the exponential
factor at $u$ needs to be averaged. After removing that factor, the
same argument applies at $v$, and then at the root.}
\label{fig:path-conditioning}
\end{figure}

\begin{proof}
The statement is immediate when $P(k)$ is empty, so assume otherwise.
We first reduce the vector bound to a scalar tail estimate by choosing a
$1/2$-net of the Euclidean unit sphere with at most $e^{2d}$ points.
For every nonzero vector $x$, a net point within distance $1/2$ of
$x/\|x\|_2$ has inner product at least $\|x\|_2/2$ with $x$.
Consequently, a union bound over the net proves
\eqref{eq:path-vector-tail}, with $C=4$, as soon as we show that every
fixed unit vector $z$ satisfies
\[
 \Pb\!\left(
   \ip{z}{\frac12\sum_{v\in P(k)}\alpha_{v,k}Y_v}
   >2\Sigma\sqrt{h(d+s)}
 \right)\le e^{-2d-s}.
\]
Indeed, a norm exceeding $4\Sigma\sqrt{h(d+s)}$ forces at least one
net direction to exceed the displayed threshold, and
$e^{2d}e^{-2d-s}=e^{-s}$. It thus remains to prove this directional
estimate.

To this end, let us fix $z$ and $\lambda>0$, and note that the path and its coefficients depend only
on the position tree and $k$, so they are deterministic. Let $u$ be
the deepest node of $P(k)$ so that every other node
$v\in P(k)\setminus\{u\}$ is a strict ancestor of $u$ (see
\cref{fig:path-conditioning}). Since $Y_v$ depends only on the tapes
at $v$ and its ancestors, conditioning on all strict-ancestor tapes
of $u$ fixes every term except the one at $u$. Its own tape remains
independent, and \eqref{eq:y-mgf}, applied with
$\lambda\alpha_{u,k}$, bounds its conditional exponential moment.
Pulling out the fixed terms and then taking expectations gives
\begin{align*}
 \mathbb E\exp\!\left(
   \frac{\lambda}{2}\sum_{v\in P(k)}\alpha_{v,k}\ip{z}{Y_v}
 \right) &=\mathbb E\!\left[
   \exp\!\left(
    \frac{\lambda}{2}\sum_{v\in P(k)\setminus\{u\}}
       \alpha_{v,k}\ip{z}{Y_v}\right)
   \mathbb E\!\left[
     e^{\lambda\alpha_{u,k}\ip{z}{Y_u}/2}
     \,\middle|\,\anc(u)
   \right]\right]\\
 &\le e^{\Sigma^2\lambda^2\alpha_{u,k}^2/2}
   \mathbb E\exp\!\left(
    \frac{\lambda}{2}\sum_{v\in P(k)\setminus\{u\}}
       \alpha_{v,k}\ip{z}{Y_v}\right).
\end{align*}
Thus averaging over the tape at $u$ removes its term from the sum,
at the cost of the deterministic factor
$e^{\Sigma^2\lambda^2\alpha_{u,k}^2/2}$. If any terms remain, the
parent of $u$ is now the deepest node, and the same calculation
applies with that parent in place of $u$. Continuing upward to the
root removes one term at a time and multiplies these factors, giving
\[
 \mathbb E\exp\!\left(
   \frac{\lambda}{2}\sum_{v\in P(k)}\alpha_{v,k}\ip{z}{Y_v}
 \right)
 \le\exp\!\left(
   \frac{\Sigma^2\lambda^2}{2}
   \sum_{v\in P(k)}\alpha_{v,k}^2\right)
 \le\exp\!\left(\frac{h\Sigma^2\lambda^2}{2}\right).
\]
The last inequality uses $0\le\alpha_{v,k}\le1$ and $|P(k)|\le h$.
Chernoff's bound with $\lambda=2\sqrt{d+s}/(\Sigma\sqrt h)$ now yields
\[
 \Pb\!\left(
   \ip{z}{\frac12\sum_{v\in P(k)}\alpha_{v,k}Y_v}
   >2\Sigma\sqrt{h(d+s)}
 \right)\le e^{-2(d+s)}\le e^{-2d-s},
\]
which is precisely the directional estimate needed above.

For the count imbalances, we use the same conditioning argument,
now with \eqref{eq:q-mgf}. For $\lambda>0$, the inequality
$e^{\lambda|x|}\le e^{\lambda x}+e^{-\lambda x}$ gives
\begin{align*}
 \mathbb E\!\left[
   e^{\lambda\alpha_{u,k}|q_u|}
   \,\middle|\,\anc(u)\right] \le\mathbb E\!\left[
   e^{\lambda\alpha_{u,k}q_u}+e^{-\lambda\alpha_{u,k}q_u}
   \,\middle|\,\anc(u)\right]
 \le 2e^{\Sigma^2\lambda^2\alpha_{u,k}^2}.
\end{align*}
Removing the nodes from deepest to root as before, and using
$2^{|P(k)|}\le e^h$ and $\sum_{u\in P(k)}\alpha_{u,k}^2\le h$, we obtain $\mathbb E\exp\!\left(\lambda\sum_{u\in P(k)}\alpha_{u,k}|q_u|\right)
 \le e^{h+h\Sigma^2\lambda^2}$. Markov's inequality with $\lambda=\sqrt{h+s}/(\Sigma\sqrt h)$ therefore
gives
\begin{align*}
 \Pb\!\left(
    \sum_{u\in P(k)}\alpha_{u,k}|q_u|
    >2\Sigma\sqrt{h(h+s)}\right) \le
 \exp\!\left(h+h\Sigma^2\lambda^2
             -2\lambda\Sigma\sqrt{h(h+s)}\right)
 =e^{-s},
\end{align*}
which proves \eqref{eq:path-count-tail}.
\end{proof}

\begin{proof}[Proof of \cref{thm:blackbox-main} (ii)]
For $T=1$ all residuals vanish, so assume $T\ge2$. Fix $k$ and use
\cref{lem:quota-repair,lem:one-path} to write
\[
 \Res_k(v_\pi)=\frac12\sum_{u\in P(k)}\alpha_{u,k}Y_u
 +\sum_{u\in P(k)}\alpha_{u,k}R_u.
\]
The first term is controlled by \eqref{eq:path-vector-tail}, while
$0\le\alpha_{u,k}\le1$ and $|P(k)|\le h$ give
\[
 \left\|\sum_{u\in P(k)}\alpha_{u,k}R_u\right\|_2
 \le\sum_{u\in P(k)}\alpha_{u,k}(|q_u|+1)
 \le h+\sum_{u\in P(k)}\alpha_{u,k}|q_u|.
\]
Thus \eqref{eq:path-count-tail} controls the second term. Apply both
estimates with $s=\log(2(T+1)/\delta)$ and take a union bound over the
$T+1$ prefixes. The total failure probability is at most
$2(T+1)e^{-s}=\delta$, so with probability at least $1-\delta$,
\begin{equation}
 \Res_{\max}(v_\pi)
 \le C\Sigma\left(\sqrt{h(d+s)}+\sqrt{h(h+s)}\right)+h.
 \label{eq:blackbox-explicit}
\end{equation}
Finally, $s \in O(h+\log(1/\delta))$ and $\Sigma\ge1$, so
\eqref{eq:blackbox-explicit} is at most
$C\Sigma\bigl(h+\sqrt{h(d+\log(1/\delta))}\bigr)$ after increasing the
absolute constant, yielding \eqref{eq:reduction-subgaussian}.
\end{proof}

\subsection{Some remarks on the reduction}
\label{sec:reduction-remarks}
\label{sec:extensions}

\paragraph{Compact convex bodies.}
\label{cor:convex-overview}
Both parts of \cref{thm:blackbox-main} extend to inputs in a known
compact convex body $K$ of affine dimension $s\ge1$. Let us consider any $0 \le k \le T$. Then, any Euclidean
online router in dimension $s$ with proportional-prefix radius $\rho$
gives, after a fixed affine normalization,
\[
 \Res_k(v_\pi)\in\frac{s\rho}{2}(K-K),
\]
with the same deterministic or probabilistic guarantee and under
the same adversarial model. Neither symmetry of $K$ nor advance
knowledge of the total is required.

Indeed, John's theorem in the relative affine space of $K$ gives
a centered ellipsoid $E\subseteq\operatorname{span}(K-K)$ and $c\in K$ with $c+E\subseteq K\subseteq c+sE$ and $2E\subseteq K-K$, where the second inclusion follows by subtracting $c-E$ from $c+E$.
Fix a linear isomorphism $\Phi$ from $\operatorname{span}(K-K)$
onto $\R^s$ with $\Phi(E)=B_2^s$, and feed
$\Phi(v_t-c)/s$ to the Euclidean router. These vectors lie in the
unit ball, and translation cancels from proportional residuals:
\[
 \sum_{j=1}^k\frac{\Phi(v_{\pi(j)}-c)}s
 -\frac{k}{T}\sum_{t=1}^T\frac{\Phi(v_t-c)}s
 =\frac1s\Phi\!\left(\Res_k(v_\pi)\right).
\]
Thus, whenever all normalized residuals have norm at most $\rho$,
\[
 \Res_k(v_\pi)\in s\rho E\subseteq\frac{s\rho}{2}(K-K).
\]
This implication holds for every realized run. Since the normalization
is fixed before the stream, it preserves online decisions and the
original deterministic or probabilistic guarantee, including its
adversarial scope. If $s=0$, every proportional residual is zero.

\paragraph{Comparison with offline transference.}
\label{rem:offline-transference}
Chobanyan's transference theorem \cite{Chobanyan1994} converts a
uniform offline prefix-signing bound $\zeta$ into the same ordering
bound. Indeed, choose an ordering of a zero-sum unit-ball sequence
minimizing its maximum prefix norm $\rho$, and sign it with prefix bound
$\zeta$. Place the positively signed vectors first in their old order,
followed by the negatively signed vectors in reverse order. Each new prefix is a
positive-sign partial sum of the old order or, by the zero total,
the negative of a negative-sign partial sum. Each is half a sum or
difference of an old unsigned prefix and its signed prefix, so
minimality gives $\rho \le (\rho+\zeta)/2$ and hence $\rho \le \zeta$. This argument chooses its initial ordering using the complete input.
Our reduction uses terminal vector and count guarantees and commits
immediately to positions, incurring losses from quota repair and
accumulation along tree paths. As shown in \Cref{sec:lower}, a loss is unavoidable already on the
line: greedy signing has bound one, whereas deterministic online
ordering requires $\Omega(\log T)$ discrepancy.

%% file: sections/03_deterministic.tex
\section{Deterministic applications}
\label{sec:deterministic}

We apply \cref{alg:tree-routing} with deterministic signing rules,
first a scalar rule and then greedy vector signing. The scalar rule
gives a stronger local invariant on the line; for greedy vector signing, we use the norm-and-count interface of \cref{thm:blackbox-main} directly.

\begin{theorem}[Deterministic OSP upper bound]
\label{thm:deterministic-main}
For every $T\ge2$, it holds that
\[
    \OS_{\det}(d,T) \in
    \begin{cases}
        O(\log T),&d=1,\\[2mm]
        O(\sqrt T),&d\ge2.
    \end{cases}
\]
These upper bounds hold against adaptive adversaries.
\end{theorem}

\subsection{Scalar routing}
\label{subsec:scalar-routing}

We begin by giving a proof of the scalar upper bound in \cref{thm:deterministic-main}.

\begin{proof}[Proof of \cref{thm:deterministic-main} for $d=1$]
We use the following local rule at every node of
\cref{alg:tree-routing}, interpreting a left proposal as $+1$
and a right proposal as $-1$.
Fix a node $u$. For the local argument, write $L,R$ for the sums
already sent to its children and $r_L,r_R$ for their remaining
capacities; once the node is full, $L=\tot_{u_L}$ and
$R=\tot_{u_R}$. We maintain
\begin{equation}
    |L+r_Ly|+|R+r_Ry|
    \le
    |L+R+(r_L+r_R)y|+2
    \qquad(y\in[-1,1]).
    \label{eq:scalar-completion-invariant}
\end{equation}
Here $L+r_Ly$ and $R+r_Ry$ are the child totals if all remaining
positions receive the same value $y$. The inequality fails exactly
when one total exceeds $1$ and the other is below $-1$. Initially
$L=R=0$, so the two totals have the same sign and the invariant holds.

If one child is full, the local rule proposes the other, so the
meta-algorithm follows that proposal. Otherwise, propose a destination
according to
\[
    \begin{cases}
        \displaystyle
        \text{left if }\frac{L+1}{r_L}\le\frac{R+1}{r_R},
        \text{ right otherwise},
        & L+R+(r_L+r_R)v_t\ge0,\\[3mm]
        \displaystyle
        \text{left if }\frac{L-1}{r_L}\ge\frac{R-1}{r_R},
        \text{ right otherwise},
        & L+R+(r_L+r_R)v_t<0.
    \end{cases}
\]
Thus the hypothetical total determines which comparison the local
rule uses; the meta-algorithm updates the chosen sum and capacity
and continues from that child.
For preservation, suppose both children have space, the hypothetical
total using $v_t$ is nonnegative, and the rule chooses left.
A subsequent completion with $y$ would give child totals
$L+v_t+(r_L-1)y$ and $R+r_Ry$. 

There are only two possible
violations of the invariant: first, suppose the new left total exceeds $1$ and the right
total is below $-1$. The left total changed by $v_t-y\le2$, so its previous value
$L+r_Ly$ must exceed $-1$. Together with $R+r_Ry<-1$, this gives
\[
    \frac{L+1}{r_L}>-y>\frac{R+1}{r_R},
\]
contradicting the choice of the smaller ratio. Second, suppose the new left total is below $-1$ and the
right total exceeds $1$.
The old invariant at $y$ requires $L+r_Ly\ge-1$, so the assignment
must have decreased this total, meaning $v_t<y$. Therefore
\[
    L+r_Lv_t
    =
    \underbrace{L+v_t+(r_L-1)y}_{<-1}
    +
    \underbrace{(r_L-1)(v_t-y)}_{\le0}
    <-1.
\]
Applying the old invariant at $v_t$ now forces $R+r_Rv_t\le1$,
and hence $L+R+(r_L+r_R)v_t<0$, contradicting the assumed
nonnegative total. Choosing right is symmetric, and the
negative-total case follows by reversing all signs.

Forced routing also preserves the invariant. Indeed, if the left child
is full, assigning $v_t$ right and then filling its remaining
positions with $y$ gives the same totals as the previous
hypothetical completion with common value
$(v_t+(r_R-1)y)/r_R\in[-1,1]$.
The old invariant applies directly; the other case is symmetric.

At completion, $r_L=r_R=0$, so the invariant becomes
\[
    |\tot_{u_L}|+|\tot_{u_R}|\le|\tot_u|+2.
\]
Since nothing in the argument requires equal child capacities, the
rule applies directly to the balanced tree for every $T$.

To conclude, fix an output prefix $[1,k]$. Starting from the root,
split only the intervals in $P(k)$, as illustrated in \cref{fig:scalar-prefix-refinement}. Each split increases the sum of absolute interval totals by at most $2$. After these splits,
the prefix is a union of complete intervals, so the triangle inequality gives
\[
    \left|\sum_{t=1}^k v_{\pi(t)}\right|
    \le |\tot_{\mathrm{root}}|+2|P(k)|
    \le 2h.
\]
The last inequality uses the zero root total and the tree height.
The invariant holds for every arriving value, so the guarantee
also holds against adaptive adversaries.
\end{proof}

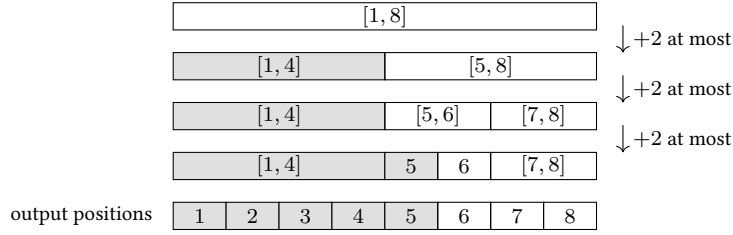
\begin{figure}[tbp]
\centering
\begin{tikzpicture}[
    x=0.70cm,y=0.60cm,
    every node/.style={font=\scriptsize},
    included/.style={fill=black!12}
]
\draw (0,3.3) rectangle (8,3.9);
\node at (4,3.6) {$[1,8]$};

\draw[included] (0,2.2) rectangle (4,2.8);
\draw (4,2.2) rectangle (8,2.8);
\node at (2,2.5) {$[1,4]$};
\node at (6,2.5) {$[5,8]$};

\draw[included] (0,1.1) rectangle (4,1.7);
\draw (4,1.1) rectangle (6,1.7);
\draw (6,1.1) rectangle (8,1.7);
\node at (2,1.4) {$[1,4]$};
\node at (5,1.4) {$[5,6]$};
\node at (7,1.4) {$[7,8]$};

\draw[included] (0,0) rectangle (4,0.6);
\draw[included] (4,0) rectangle (5,0.6);
\draw (5,0) rectangle (6,0.6);
\draw (6,0) rectangle (8,0.6);
\node at (2,0.3) {$[1,4]$};
\node at (4.5,0.3) {$5$};
\node at (5.5,0.3) {$6$};
\node at (7,0.3) {$[7,8]$};

\foreach \y in {3.3,2.2,1.1}
    \draw[->] (8.5,\y+0.1)
        --node[right] {$+2$ at most}(8.5,\y-0.5);

\node[anchor=east] at (-0.2,-0.8) {output positions};
\foreach \j in {1,...,8}{
    \ifnum\j<6
        \fill[black!12] (\j-1,-1.1) rectangle (\j,-0.5);
    \fi
    \draw (\j-1,-1.1) rectangle (\j,-0.5);
    \node at (\j-0.5,-0.8) {$\j$};
}
\end{tikzpicture}
\caption{Analyzing $[1,5]$ when $T=8$. Only a partially included
interval is split; shaded blocks are wholly inside the prefix.
Each refinement costs at most $2$ in the sum of absolute block
totals, so three refinements bound this prefix by $6$.}
\label{fig:scalar-prefix-refinement}
\end{figure}

\subsection{Greedy signing in arbitrary dimension}
\label{subsec:det-vector-routing}

The same routing algorithm, using a greedy signer at each node,
controls proportional-prefix residuals even when the total is nonzero. Hence, the upper bound in \Cref{thm:deterministic-main} for $d \ge 2$ immediately follows since, for zero-sum inputs, the residuals are
ordinary prefix sums.

\begin{proposition}[Deterministic proportional routing]
\label{prop:det-router}
For every stream of Euclidean unit-ball vectors, \cref{alg:tree-routing} with greedy signing of $(v_t,1)/\sqrt2$ satisfies
\begin{equation}
    \Res_{\max}(v_\pi)
    \le \frac{\sqrt2+1}{1-\sqrt{2/3}}\sqrt T.
    \label{eq:det-router}
\end{equation}
This bound holds against adaptive adversaries and requires no
advance knowledge of the total.
\end{proposition}

\begin{proof}
At an internal node $u$, let $v_{u,i}$ be its $i$th arriving vector
and let $S_{u,i}$ be the signer's augmented signed sum after that
arrival, initially $S_{u,0}=0$. Choose
$\eps_{u,i}\in\{-1,+1\}$ so that
\[
    \eps_{u,i}\left\langle
        S_{u,i-1},\frac{(v_{u,i},1)}{\sqrt2}
    \right\rangle\le0,
\]
breaking ties in favor of $+1$, and update
$S_{u,i}=S_{u,i-1}+\eps_{u,i}(v_{u,i},1)/\sqrt2$.
The update records the proposed sign even if the algorithm overrides
its routing request. Expanding the squared norm gives
\[
\begin{aligned}
    \|S_{u,i}\|_2^2
    =\|S_{u,i-1}\|_2^2
      +\frac{\|v_{u,i}\|_2^2+1}{2}
      +2\eps_{u,i}\left\langle
          S_{u,i-1},\frac{(v_{u,i},1)}{\sqrt2}
       \right\rangle
    \le\|S_{u,i-1}\|_2^2+1\le i.
\end{aligned}
\]
Thus the vector and count imbalances obey the joint bound
\[
    \left\|\sum_{i=1}^{m_u}\eps_{u,i}v_{u,i}\right\|_2^2
    +\left|\sum_{i=1}^{m_u}\eps_{u,i}\right|^2
    =2\|S_{u,m_u}\|_2^2\le2m_u.
\]
The local estimate in the norm-and-count part of
\cref{thm:blackbox-main}, followed by Cauchy--Schwarz, therefore gives
\[
    \|\err_u\|_2
    \le\frac12\left(
        \left\|\sum_{i=1}^{m_u}\eps_{u,i}v_{u,i}\right\|_2
        +\left|\sum_{i=1}^{m_u}\eps_{u,i}\right|
       \right)+\1\{m_u\text{ odd}\} \le\sqrt{m_u}+\1\{m_u\text{ odd}\}.
\]
The boundary-path conclusion of \cref{thm:blackbox-main} now yields
\[
    \|\Res_k(v_\pi)\|_2 \le\sum_{u\in P(k)}\|\err_u\|_2 \le(\sqrt2+1)\sqrt T
       \sum_{\ell=0}^{\infty}(2/3)^{\ell/2} =\frac{\sqrt2+1}{1-\sqrt{2/3}}\sqrt T.
\]
Indeed, every internal node has children of capacity at most
$2m_u/3$, so a node at depth $\ell$ has capacity at most
$T(2/3)^\ell$, and $P(k)$ contains at most one node at each depth.
Every estimate holds for every local stream, which also proves
the adaptive guarantee.
\end{proof}

%% file: sections/04_corollaries.tex
\section{Randomized applications}
\label{sec:corollaries}

The reduction in \cref{thm:blackbox-main} accepts any online signer
satisfying its terminal guarantees, and, as in the deterministic case, this flexibility yields several
randomized applications. As our main application, we use Aden-Ali's signer \cite{AdenAli2026} since it
satisfies the subgaussian hypothesis in part~\textup{(ii)} and can be implemented efficiently. Moreover, we will derive a coordinatewise bound from it, and a sparse-vector application from the norm-and-count guarantees in
part~\textup{(i)} and the Altschuler--Tikhomirov's signer \cite{AltschulerTikhomirov2026}. Throughout this section, inputs are fixed independently of the algorithm's randomness, and all bounds allow
arbitrary, initially unknown totals and become prefix-discrepancy
bounds when the total is zero.

\subsection{Efficient randomized ordering}

\begin{corollary}[Efficient signer]
\label{cor:aa-overview}
For every $T\ge2$, $d\ge1$, and $0<\delta<1$, \Cref{alg:tree-routing}
with the Aden-Ali signer uses $O(dT\log T)$ arithmetic operations and,
for every fixed Euclidean unit-ball stream, satisfies, with probability at least $1-\delta$,
\[
 \Res_{\max}(v_\pi)
 \in O\left(\log T+
 \sqrt{\bigl(d+\log(1/\delta)\bigr)\log T}\right).
\]
\end{corollary}

\begin{proof}
We apply part~\textup{(ii)} of \cref{thm:blackbox-main}, whose assumption is the terminal moment-generating-function bound
\eqref{eq:terminal-mgf} in dimension $d+1$. Specifically, at a node receiving
$x_1,\ldots,x_m$, we run an independent copy of Aden-Ali's signer on $z_i = (x_i,1)/\sqrt{2}$ so that $\|z_i\|_2^2 \le 1$. We keep this copy running throughout the local stream, recording its proposed signs even when the router overrides a destination to
enforce the child capacities.

To verify the required assumption in \eqref{eq:terminal-mgf}, we observe that, for every fixed augmented stream, Aden-Ali's signer's guarantee
\cite[Theorem~1 and Equation~(5)]{AdenAli2026} represents the terminal
signed sum as $\sum_{i=1}^m\eps_i z_i=G_1+G_2+G_3$, where $G_1,G_2,G_3 \sim N(0,I_{d+1})$ are Gaussian vectors that may be dependent. Then, for every
$\theta\in\R^{d+1}$, H\"older's inequality and the Gaussian
moment-generating function give
\begin{equation}
\begin{aligned}
 \mathbb E\exp\!\left(\ip{\theta}{\sum_{i=1}^m\eps_i z_i}\right)
 &\le\prod_{j=1}^3
   \left(\mathbb E e^{3\ip{\theta}{G_j}}\right)^{1/3} =\prod_{j=1}^3
   \left(e^{9\norm{\theta}^2/2}\right)^{1/3}
 =e^{9\norm{\theta}^2/2}.
\end{aligned}
\label{eq:aa-mgf}
\end{equation}
Hence, \eqref{eq:terminal-mgf} is verified with $\Sigma=3$, uniformly over all node capacities.
Furthermore, conditioning on a node's strict-ancestor tapes fixes
its input stream while leaving its own tape independent, and thus,
\eqref{eq:aa-mgf} also supplies the conditional hypothesis used in
\cref{lem:quota-repair,lem:path-composition}.

For completeness, we spell out the resulting substitutions in the
reduction: set $h=\lceil\log_2T\rceil$ and
$s=\log(2(T+1)/\delta)$, and, with $\Sigma=3$, the two estimates of
\cref{lem:path-composition} read, for each $k$,
\begin{align*}
 \Pb\!\left(
 \left\|\frac12\sum_{u\in P(k)}\alpha_{u,k}Y_u\right\|_2
 >3C\sqrt{h(d+s)}\right)&\le e^{-s},\\
 \Pb\!\left(
 \sum_{u\in P(k)}\alpha_{u,k}|q_u|
 >3C\sqrt{h(h+s)}\right)&\le e^{-s}.
\end{align*}
A union bound over these two events and the $T+1$ prefixes has total
failure probability at most $2(T+1)e^{-s}=\delta$. Outside this
failure event, \cref{lem:one-path,lem:quota-repair} give, for every $k$,
\begin{align*}
 \norm{\Res_k(v_\pi)}
 &\le\left\|\frac12\sum_{u\in P(k)}\alpha_{u,k}Y_u\right\|_2
   +\sum_{u\in P(k)}\alpha_{u,k}\norm{R_u} \le\left\|\frac12\sum_{u\in P(k)}\alpha_{u,k}Y_u\right\|_2
   +h+\sum_{u\in P(k)}\alpha_{u,k}|q_u|\\
 &\le 3C\sqrt{h(d+s)}+h+3C\sqrt{h(h+s)} \in O\left(\log T+
 \sqrt{\bigl(d+\log(1/\delta)\bigr)\log T}\right),
\end{align*}
where the last step uses
$s=O(h+\log(1/\delta))$, $\sqrt{a+b}\le\sqrt a+\sqrt b$ and $h \in O(\log T)$.

Finally, a length-$m$ signer in dimension $d+1$ uses
$O(dm)$ arithmetic operations, and, since at each depth, the node
capacities sum to at most $T$, and there are at most $h$ internal
levels, the overall running time
$O(dT\log T)$.
\end{proof}

We next discuss an alternative signer and establish a sharper dimension dependence.

\paragraph{Using an inefficient signer.}  If computational efficiency is not required, the reduction also
allows the Kulkarni--Reis--Rothvoss's signer in place of Aden-Ali's.
For each fixed augmented unit-ball stream, its terminal guarantee is
\cite[Theorem~1]{KulkarniReisRothvoss2024}
\begin{equation}
 \mathbb E\exp\!\left(
 \frac{\ip{\theta}{\sum_i\eps_i z_i}^2}{100\norm{\theta}^2}
 \right)\le2
 \qquad(\theta\ne0).
 \label{eq:krr-psi2}
\end{equation}
At each node, sample an independent uniform sign $\tau\in\{-1,+1\}$
before processing its stream, and send $\tau\eps_i$ as the proposals
without changing the Kulkarni--Reis--Rothvoss signer's internal updates. The terminal sum
is then symmetric and centered, while \eqref{eq:krr-psi2} is
unchanged. Accordingly, we call this the \emph{symmetrized Kulkarni--Reis--Rothvoss's signer}. By \cite[Lemma~11]{KulkarniReisRothvoss2024},
\[
 \mathbb E\exp\!\left(\ip{\theta}{\tau\sum_i\eps_i z_i}\right)
 \le e^{40\norm{\theta}^2}.
\]
Thus \eqref{eq:terminal-mgf} holds with $\Sigma=\sqrt{80}$, and
\cref{thm:blackbox-main}\textup{(ii)} gives the same probability and
discrepancy bound as \cref{cor:aa-overview}, up to the absolute
constant. Only the efficiency guarantee is lost, since the finite-horizon
Kulkarni--Reis--Rothvoss's construction takes exponential time
\cite[Sections~6--7]{KulkarniReisRothvoss2024}.

\begin{lemma}[Sharper dimension dependence]
\label{lem:stopped-tree}
For every $T\ge2$, $d\ge1$, and $0<\delta<1$, there is a randomized
router such that, for every fixed Euclidean unit-ball stream,
with probability at least $1-\delta$,
\begin{equation}\label{eq:stopped-tree}
 \Res_{\max}(v_\pi)
 \in O\left(\min\Bigl\{\sqrt T,\;
 \log T+\sqrt{d_*\log(2+T/d_*)}
 +\sqrt{\log(2+T/d_*)\log(1/\delta)}\Bigr\}\right),
\end{equation}
where $d_*:=\min\{d,T\}$.
With the Aden-Ali signer, the running time is $O(dT\log T)$.
The same bound holds with the symmetrized Kulkarni--Reis--Rothvoss's signer, with a
different absolute constant and without the running-time guarantee.
\end{lemma}

\begin{proof}
Use the randomized signer only at nodes strictly above a depth
$0\le s\le h=\lceil\log_2T\rceil$, and the deterministic router of
\cref{prop:det-router} within each remaining interval.
For $s\ge1$, apply \cref{lem:path-composition} with path length
at most $s$ and tail parameter $\log(2^{s+1}/\delta)$.
A union bound over the two estimates at each of the at most
$2^s-1$ nontrivial interval endpoints bounds all their residuals by
$C(s+\sqrt{s(d+\log(1/\delta))})$ with probability at least $1-\delta$.
Inside an interval, a global residual equals the local proportional
residual plus a convex combination of the two endpoint residuals.
Since each interval has length at most $\lceil T/2^s\rceil$,
\cref{prop:det-router} therefore gives
\[
 \Res_{\max}(v_\pi)
 \in O\left(s+\sqrt{s\bigl(d+\log(1/\delta)\bigr)}
 +\sqrt{T/2^s}\right)
\]
on the same event, where the ceiling is absorbed into the constant
because $s\le h$.

For $s=0$, this is the deterministic $C\sqrt T$ bound.
When $d<T$, choose $s=\lceil\log_2(T/d)\rceil$.
Then $T/2^s\le d$ and $s=O(\log(2+T/d))$.
Comparing this choice with $s=0$ before the stream begins gives
\eqref{eq:stopped-tree}.
For $d\ge T$, use $s=0$.
The running time remains $O(dT\log T)$ with the Aden-Ali signer,
since the total node capacity at each depth is at most $T$.
The same probability argument applies to the symmetrized Kulkarni--Reis--Rothvoss's signer.
Finally, if $\delta\ge T^{-\beta}$ for fixed $\beta>0$, then $\sqrt{\log(2+T/d_*)\log(1/\delta)}=O_\beta(\log T).$
\end{proof}

\subsection{Coordinatewise discrepancy}

The next consequence uses the directional estimates in the proof of
part~\textup{(ii)} of \cref{thm:blackbox-main}. A union bound over
the relevant coordinate directions replaces the sphere net used for
the $\ell_2$-norm.

\begin{corollary}[Coordinatewise ordering]
\label{rem:coordinatewise}
For every fixed Euclidean unit-ball stream and $0<\delta<1$,
\Cref{alg:tree-routing}, using either the symmetrized Kulkarni--Reis--Rothvoss's
signer or the Aden-Ali signer throughout the full tree, satisfies, with probability at least $1-\delta$,
\[
 \max_{0\le k\le T}\|\Res_k(v_\pi)\|_\infty
 \in O\left(\log T + \sqrt{\log T \cdot \log(1/\delta)}\right).
\]
\end{corollary}

\begin{proof}
The proof of
\cref{cor:aa-overview} verifies
\eqref{eq:terminal-mgf} with an absolute scale $\Sigma$ for either signer. We use the same path decomposition and
quota-repair estimate as in \cref{thm:blackbox-main}, but first
identify which coordinates need a probabilistic estimate. By Cauchy--Schwarz and the unit-ball assumption,
\[
 \sum_{j=1}^d\left(\sum_{t=1}^T|v_{t,j}|\right)^2
 \le T\sum_{t=1}^T\|v_t\|_2^2\le T^2,
\]
which implies that at most $T^2$ coordinates satisfy
$\sum_t|v_{t,j}|>1$. Every other coordinate has residual of magnitude
at most one for every permutation: in a proportional residual, each
input coordinate has coefficient either $1-k/T$ or $-k/T$, both of
absolute value at most one. It is worth noting that the exceptional coordinates depend only
on the fixed input and are used only in the analysis, and the algorithm
does not need to identify them.

For any fixed unit direction $z$, the conditional moment calculation
in the proof of \cref{lem:path-composition}, followed by Chernoff's
bound, gives, for all $s > 0$,
\[
 \Pb\!\left(
 \frac12\sum_{u\in P(k)}\alpha_{u,k}\ip{z}{Y_u}
 >\Sigma\sqrt{2hs}\right)\le e^{-s}.
\]
Apply this estimate to both signs of each exceptional coordinate
direction. Also apply the count estimate \eqref{eq:path-count-tail}
at each prefix. There are at most $(2T^2+1)(T+1)$ events in total,
so choosing $s = \log((2T^2+1)(T+1)/\delta)$ and taking a union bound makes them all hold with probability at least $1-\delta$.

On this event, the centered path sum in every exceptional coordinate
has magnitude at most $\Sigma\sqrt{2hs}$. The correction is controlled
without a union bound over coordinates, because
\[
 \left\|\sum_{u\in P(k)}\alpha_{u,k}R_u\right\|_\infty
 \le\sum_{u\in P(k)}\alpha_{u,k}\|R_u\|_2
 \le h+\sum_{u\in P(k)}\alpha_{u,k}|q_u|.
\]
Combining the two estimates gives, at every prefix and in every
exceptional coordinate, a residual of magnitude at most
\[
 \Sigma\sqrt{2hs}+h+C\Sigma\sqrt{h(h+s)}
 \le C\Sigma\bigl(h+\sqrt{hs}\bigr).
\]
Now $s \in O(\log T+\log(1/\delta))$ and
$h \in O(\log T)$. Absorbing the bound one for the remaining
coordinates proves the claim. The running-time calculation is the
same as in \cref{cor:aa-overview}.
\end{proof}

\subsection{Sparse vectors}

Here we apply part~\textup{(i)} of \cref{thm:blackbox-main} in the
$\ell_\infty$ norm, with the input vectors that need only be coordinate-bounded, and their $\ell_2$-norms may be as large as $\sqrt\kappa$. In particular, we use copies of the Altschuler--Tikhomirov signer \cite{AltschulerTikhomirov2026}.

\begin{corollary}[OSP for sparse vectors]
\label{cor:at-overview}
For every $\eta>0$, there is a constant $C_\eta>0$ with the following
property. Let $T\ge2$, let $\kappa\in\{1,\ldots,d\}$ be a sparsity
bound known in advance, and let $0<\delta<1$. Write
$h=\lceil\log_2T\rceil$ and suppose that $\delta \ge C_\eta Th\exp(-(\kappa+1)/(C_\eta\log^{2+\eta}(e(\kappa+1))))$. Then, for every fixed stream $v_t\in[-1,1]^d$ with
$\|v_t\|_0\le\kappa$, \Cref{alg:tree-routing}, using the
Altschuler--Tikhomirov signer, satisfies, with probability at least $1-\delta$,
\[
 \max_{0\le k\le T}\|\Res_k(v_\pi)\|_\infty
 \le(C_\eta\sqrt{\kappa+1}+1)h \in O_\eta(\sqrt\kappa\log T).
\]
\end{corollary}

\begin{proof}
Fix $\kappa\ge1$ and $T\ge2$. At a node receiving
$x_1,\ldots,x_m$, append the quota coordinate and give $(x_i,1)$
to an independent copy of the signer of Altschuler--Tikhomirov
\cite[Theorem~1.1]{AltschulerTikhomirov2026}, with announced sparsity
$\kappa+1$. Thus, each augmented vector belongs to $[-1,1]^{d+1}$ and has
at most $\kappa+1$ nonzero coordinates. Moreover, note that no Euclidean normalization
is needed here, because we are using the norm-and-count part of the
reduction with the $\ell_\infty$ norm. Next, for each fixed local stream, the signer's simultaneous-prefix
guarantee implies, at the terminal prefix,
\[
 \left\|\sum_{i=1}^m\eps_i x_i\right\|_\infty
 \le C_\eta\sqrt{\kappa+1},
 \qquad
 \left|\sum_{i=1}^m\eps_i\right|
 \le C_\eta\sqrt{\kappa+1},
\]
except with probability at most
\[
 p(m)=C_\eta m\exp\!\left(
 -\frac{\kappa+1}{C_\eta\log^{2+\eta}(e(\kappa+1))}
 \right).
\]
We have enlarged $C_\eta$, if necessary, so that the same constant
works in both estimates. The first terminal bound comes from the
original coordinates, and the second comes from the appended quota
coordinate. Thus the two hypotheses of part~\textup{(i)} of
\cref{thm:blackbox-main} hold jointly with
$V(m)=W(m)=C_\eta\sqrt{\kappa+1}$.

Although the stream reaching a node is random, it is determined by
its strict-ancestor tapes. Conditional on those tapes, the node's
own signer is an independent copy on a fixed stream: this is the
conditioning argument in \cref{thm:blackbox-main}, so its failure
bound applies at every node. At each depth, the node capacities sum
to at most $T$, and there are at most $h$ internal levels, so
\[
 \sum_{u\text{ internal}}p(m_u)
 \le C_\eta Th\exp\!\left(
 -\frac{\kappa+1}{C_\eta\log^{2+\eta}(e(\kappa+1))}
 \right)\le\delta.
\]
With probability at least $1-\delta$, all local terminal bounds
hold. On this single event, \eqref{eq:reduction-norm} gives, for
every output prefix,
\[
 \|\Res_k(v_\pi)\|_\infty
 \le\sum_{u\in P(k)}
   \left(C_\eta\sqrt{\kappa+1}+\1\{m_u\text{ odd}\}\right)
 \le(C_\eta\sqrt{\kappa+1}+1)h.
\]
This gives
the stated asymptotic bound as $h \in O(\log T)$.
\end{proof}

%% file: sections/05_lower_bounds.tex
\section{Lower bounds}
\label{sec:lower}
In this section, we prove lower bounds that match our deterministic
upper bounds and, for constant and polynomially small failure probabilities, our randomized upper bounds.
All our hard inputs use the same zero-sum completion: stop an initial stream and fill the remaining arrivals with identical vectors that
cancel its total. 
For the randomized bounds, the initial stream consists of coordinate cycles, stopped at a time chosen independently of the
algorithm's random tape. 
For the deterministic bound in a plane, the initial stream instead depends on the algorithm's previous assignments.

\begin{theorem}[Deterministic and randomized OSP lower bounds]
\label{thm:random-lower-overview}
\label{lem:scalar-extraction}
\label{lem:scalar-tail}
\label{lem:det-lower-bound-higher-dim}
For $T\ge2$ and $d\ge1$, it holds that
\[
 \OS_{\det}(d,T)\in
 \begin{cases}
 \Omega(\log T),&d=1,\\
 \Omega(\sqrt T),&d\ge2,
 \end{cases}
\]
and, for $0<\delta\le1/2$ and $d_*:=\min\{d,T\}$,
\[
 \OS_{\rm rand}(d,T,\delta)
 \in \Omega\left(\log T+\sqrt{d_*\log(T/d_*)}\right).
\]
\end{theorem}

The deterministic bound for $d\ge2$ already holds in a fixed plane.
The oblivious hard inputs are independent of the algorithm's random tape.

\subsection{The input construction}
We first explain the construction of the hard inputs which will be used throughout this section.
Fix $T\ge2$. Consider the following input: after presenting $m\le T/2$ vectors $v_1,\ldots,v_m$ with $\norm{v_i}\le1$, present $(T-m)$ copies of $-1/(T-m)\sum_{i=1}^m v_i$.
These repeated vectors remain in the unit ball since $\norm{\sum_{i=1}^m v_i}\le m\le T-m$.
Let $p_i \in [T]$ be the position assigned to $v_i$ by the algorithm. Then the first $k$ output positions contain exactly $k-\lvert\{i\le m:p_i\le k\}\rvert$ repeated vectors, hence the prefix residual is
\begin{equation}
 \Res_k(v_\pi)=
 \sum_{\substack{i\le m\\p_i\le k}}v_i
 -\frac{k-\lvert\{i\le m:p_i\le k\}\rvert}{T-m}
       \sum_{i=1}^m v_i.
 \label{eq:scalar-prefix-identity}
\end{equation}
For the randomized lower bounds, suppose that $d\le T/2$ and set $N=\lfloor T/(2d)\rfloor$.
For each $s\in[N]$, the initial stream consists of $s$ cycles through the coordinate vectors, giving the fixed input
\begin{equation}
 I_s=\left(
 \underbrace{e_1,\ldots,e_d,\ldots,e_1,\ldots,e_d}_{s\text{ cycles}},
 \underbrace{-\frac{s\mathbf1}{T-ds},\ldots,
             -\frac{s\mathbf1}{T-ds}}_{T-ds\text{ copies}}\right),
 \label{eq:scalar-family}
\end{equation}
where $\mathbf1=(1,\ldots,1)\in\R^d$.
The point is that the algorithm cannot know when the coordinate cycles will stop.
Indeed, fix a random tape and let $p_{i,j}$ be the position assigned to $e_j$ in cycle $i$ when the algorithm runs on $I_N$. Clearly, these same positions describe its first $s$ cycles
on every $I_s$, since the inputs agree up to time $ds$.

\subsection{Deterministic lower bounds}
We prove a logarithmic lower bound in dimension one and a
square-root lower bound in every dimension at least two.
\begin{lemma}[Deterministic lower bounds]
\label{lem:deterministic-lower}
There exists an absolute constant $c>0$ such that, for every $T\ge2$
and $d\ge1$, every deterministic online ordering algorithm has a
fixed zero-sum unit-ball input satisfying
\[
 \Res_{\max}(v_\pi)\ge
 \begin{cases}
   c\log T, & d=1,\\
   c\sqrt T, & d\ge2.
 \end{cases}
\]
For $d=1$, the input can be chosen from the family
\eqref{eq:scalar-family}; for $d\ge2$, the input lies in a fixed plane.
\end{lemma}

\begin{proof}
On the line, let us take $N=\lfloor T/2\rfloor$ and write
$p_i=p_{i,1}$ for the common positions in the construction above.
Let $c_0>0$ be the absolute constant in Schmidt's
theorem~\cite[Theorem~2]{Schmidt1972}. Applied to the points
$((i-1)/N,(p_i-1)/T)$, it gives $x,y\in[0,1]$ such that
\[
 \left|
 \left|\left\{
 i\le N:
 \frac{i-1}{N}<x,\ 
 \frac{p_i-1}{T}<y
 \right\}\right|
 -Nxy
 \right|
 \ge c_0\log N.
\]
Here the cardinality counts the points in $[0,x)\times[0,y)$,
while $Nxy$ is $N$ times the area of this rectangle. We now set $s=\lceil Nx\rceil$ and $k=\lceil Ty\rceil$, and note that rounding the upper sides to $s/N$ and $k/T$ preserves the point count because $(i-1)/N < x$ if and only if $i \le s$ and $(p_i-1)/T < y$ if and only if $p_i \le k$. 
In addition, this changes the area term by at most
\[
 0\le \frac{sk}{T}-Nxy
 =
 N\left(\frac{s}{N}-x\right)y
 +s\left(\frac{k}{T}-y\right)
 \le 1+\frac{N}{T}\le\frac32.
\]
The reverse triangle inequality therefore gives
\[
 \left|
 \lvert\{i\le s:p_i\le k\}\rvert-\frac{sk}{T}
 \right|
 \ge c_0\log N-\frac32
 \ge c_0\log N-2.
\]
On the input $I_s$, the completion identity
\eqref{eq:scalar-prefix-identity} becomes
\begin{equation}
 \Res_k(v_\pi)
 =
 \frac{T}{T-s}
 \left(\lvert\{i\le s:p_i\le k\}\rvert-\frac{sk}{T}\right).
 \label{eq:scalar-lower-cost}
\end{equation}
Since $T/(T-s)\ge1$, this proves the logarithmic lower bound
for all sufficiently large $T$.

For the planar bound, we proceed in a similar way to Spencer \cite{Spencer77,Spencer94}, and present the proof for completeness. Indeed, let us assume $T\ge8$, take
$m=\lfloor T/4\rfloor$, and focus on the output prefix
$k=\lfloor T/2\rfloor$.
During the first $m$ arrivals, let $\xi_t=+1$ if the algorithm assigns
$v_t$ to one of the first $k$ positions, and $\xi_t=-1$ otherwise,
and write $U_t=\sum_{i=1}^t v_i, Z_t=\sum_{i=1}^t \xi_i v_i$ with $U_0=Z_0=0$.
We choose the vectors so that the difference $Z_t$ between the
two groups grows, while their total $U_t$ stays small enough that
the completion cannot cancel this difference. At time $t$, choose a unit vector $v_t$ in the plane spanned by
$e_1,e_2$ perpendicular to $Z_{t-1}$, reversing its orientation
if necessary so that $\ip{U_{t-1}}{v_t}\le0$.
This choice uses only previous assignments, and whichever value
of $\xi_t$ the algorithm then produces,
\[
 \norm{Z_t}^2=\norm{Z_{t-1}}^2+1,
 \qquad
 \norm{U_t}^2\le\norm{U_{t-1}}^2+1.
\]
Hence $\norm{Z_m}=\sqrt m$ and $\norm{U_m}\le\sqrt m$. We now complete the stream with $T-m$ copies of $-U_m/(T-m)$.
Among the initial vectors, those assigned to the first $k$
positions have sum $(U_m+Z_m)/2$ and number
$(m+\sum_{t=1}^m \xi_t)/2$, so
\eqref{eq:scalar-prefix-identity} gives
\[
 2\Res_k(v_\pi)
 =
 Z_m+\frac{T-2k+\sum_{t=1}^m \xi_t}{T-m}U_m.
\]
Since $T-2k\in\{0,1\}$ and
$\lvert\sum_{t=1}^m \xi_t\rvert\le m$, the triangle inequality yields
\[
 2\norm{\Res_k(v_\pi)}
 \ge
 \left(1-\frac{m+1}{T-m}\right)\sqrt m
 \ge \frac{\sqrt m}{2}
 \ge \frac{\sqrt T}{4\sqrt2},
\]
where $(m+1)/(T-m)\le1/2$ and $m\ge T/8$ for $T\ge8$.
Choosing $c>0$ sufficiently small gives both bounds for large $T$;
the remaining bounded values follow from the one-cycle bound below
after further decreasing $c$. All vectors lie in the same fixed plane, and although the initial stream was described adaptively, for any deterministic algorithm its assignments can be simulated in advance, producing a fixed input with the claimed cost.
\end{proof}

\begin{remark}
    It is worth noting that the one-cycle input $I_1$, using $\min\{d,\lfloor T/2\rfloor\}$ coordinates, has discrepancy at least $c\sqrt{\min\{d,T\}}$ under every ordering, for an absolute constant $c>0$. Indeed, consider the one-cycle input $I_1$ in dimension $d\le T/2$.
For any ordering, linearly interpolate between consecutive prefix
sums and stop during a repeated vector, or at its endpoint, when
exactly half of the repeated copies have been traversed.
Their contribution is then $-1/2$ in each coordinate, while
each basis vector has either appeared in full or has not appeared
yet, so every coordinate of the interpolated sum is $-1/2$ or $1/2$.
Its norm is therefore $\sqrt d/2$, and convexity of the norm implies
that at least one adjacent genuine prefix has norm at least as large.
\end{remark}

\subsection{Randomized lower bounds}
We now show that randomized algorithms can also have a large prefix sum on some fixed input. 
We start by proving the first term $\Omega(\log T)$ of the randomized lower bound, which holds already in dimension one.
\begin{lemma}[Randomized lower bound for $d=1$]
\label{lem:random-scalar-lower}
There exists an absolute constant $c>0$ such that, for every $T\ge2$ and every randomized online ordering algorithm, some
fixed input from the scalar family \eqref{eq:scalar-family}
satisfies
\[
 \Pb\!\left(\Res_{\max}(v_\pi)>c\log T\right)>\frac{1}{2}.
\]
\end{lemma}

\begin{proof}
Let $d=1$ and $N=\lfloor T/2\rfloor$. The family \eqref{eq:scalar-family} becomes
\[
 I_s=\left(
 \underbrace{1,\ldots,1}_{s\text{ copies}},
 \underbrace{-\frac{s}{T-s},\ldots,-\frac{s}{T-s}}_{T-s\text{ copies}}
 \right),
 \qquad s\in[N].
\]
Let us write $\Res_{\max}(I_s)$ for the maximum residual of the algorithm's output on $I_s$, and let us fix a random tape so that the algorithm's decisions are deterministic, with $p_i$ being the output position assigned to the $i$th arriving copy of $1$ on $I_N$.
Using the same tape, its first $s$ assignments on $I_s$ are also $p_1,\ldots,p_s$, since both inputs begin with $s$ copies of $1$.
For $s\in[N]$, denote the discrepancy over output prefixes as
\[
 \Delta_s:=
 \max_{0\le k\le T}
 \left|
   \lvert\{i\le s:p_i\le k\}\rvert-\frac{sk}{T}
 \right|.
\]
By \eqref{eq:scalar-lower-cost}, we have
$\Delta_s=(T-s)/T \cdot \Res_{\max}(I_s)\le\Res_{\max}(I_s)$.
To apply Tijdeman--Wagner's theorem on an infinite sequence in $[0,1)$, we first append infinitely many zeros after $((p_i-1)/T)_{i=1}^N$.
By Tijdeman--Wagner's theorem~\cite[Theorem~1]{TW80}, there exists an absolute constant $c_0>0$ such that for every fixed tape,
\[
 \left|\left\{1\le s\le N:\Delta_s+\frac{1}{2}<c_0\log N\right\}\right|
 \le 2N^{5/6}.
\]
Here the term $1/2$ comes from restricting interval endpoints to multiples of $1/T$.
Now, for any $y\in[0,1)$, we round $y$ upwards to $\lceil Ty\rceil/T$, and, since all points lie on this grid, the intervals $[0,y)$ and $[0,\lceil Ty\rceil/T)$ contain the same points.
The term $sy$ changes by at most $s/T \leq 1/2$,
so the discrepancy over arbitrary intervals $[0,y)$ is at most $\Delta_s+1/2$.

Let $c_0$ be the constant in Tijdeman--Wagner's theorem and set
$c_1=\min\{c_0/4,1/26\}$. If $c_1\log T<1/2$, the input $I_1$
gives the desired bound with probability one. Otherwise,
$\log T\ge 13$, so $N>2^{12}$ and $\log N>1/2 \cdot \log T$, giving $c_1\log T+ 1/2 \le 2c_1\log T <c_0\log N$. This gives that, for any $s$ with $\Res_{\max}(I_s)\le c_1\log N$,
\[
 \Delta_s+\frac{1}{2}
 \le \Res_{\max}(I_s)+\frac{1}{2}
 \le c_1\log N+\frac{1}{2}
 <c_0\log N.
\]
Therefore for each random tape, there are at most
$2N^{5/6}$ indices with $\Res_{\max}(I_s)\le c_1\log N$.
Averaging over the algorithm's random tape gives
\[
 \frac{1}{N}\sum_{s=1}^N
 \Pb\!\left(\Res_{\max}(I_s)>c_1\log N\right)
 \ge 1-2N^{-\frac{1}{6}}>\frac{1}{2},
\]
so there exists a fixed $I_s$ which has maximum residual greater than $c\log N$ with probability greater than $1/2$. Hence, some fixed
input $I_s$ satisfies
\[
 \Pb\!\left(\Res_{\max}(I_s)>c\log T\right)>\frac12,
\]
for $c = c_1/2$, which concludes the proof.
\end{proof}
The following lemma derives the second term of the randomized lower bound for any dimension.
\begin{lemma}[Randomized lower bound in arbitrary dimension]
\label{lem:random-dimension-lower}
There exists an absolute constant $c>0$ such that the following holds.
Let $T\ge2$, $d\ge1$, and $0<\delta<1$.
If a randomized online ordering algorithm satisfies
\[
 \Pb\!\left(\Res_{\max}(v_\pi)\le \rho\right)\ge1-\delta
\]
on every fixed zero-sum unit-ball input, then $\rho\ge
 c(1-\delta)\sqrt{d_*\log(T/d_*)}$, where $d_*:=\min\{d,T\}$.
\end{lemma}

\begin{proof}
When $d>T/2$, the one-cycle input in
\cref{lem:deterministic-lower} has cost $\Omega(\sqrt{d_*})$
under every ordering and, since $\log(2+T/d_*) \in O(1)$, this proves
the claim for sufficiently small $c>0$.

Therefore we assume that $d\le T/2$ and let $N=\lfloor T/(2d)\rfloor$.
The expected fraction of inputs $I_1,\ldots,I_N$ on which the
algorithm succeeds is at least $1-\delta$, and so we fix
a tape that succeeds on a set $S\subseteq[N]$ with probability $p:=|S|/N\ge1-\delta$.
For its positions $p_{i,j}$, define for each $j=1,\ldots,d$ and $(x,y)\in[0,1]^2$,
\[
 D_j(x,y):=
 \left|\left\{
 i\in[N]:
 \frac{i-1}{N}<x,\ 
 \frac{p_{i,j}-1}{T}<y
 \right\}\right|-Nxy,
\]
and for fixed $s\in[N]$ and $0\le k\le T$, write $a_j:=D_j\left(s/N,k/T\right) =\lvert\{i\le s:p_{i,j}\le k\}\rvert-sk/T$. The identity in \eqref{eq:scalar-prefix-identity} can thus be rewritten as
\[
 \Res_k(v_\pi)
 =(a_j)_{j=1}^d
 +\frac{s}{T-ds}\left(\sum_{j=1}^d a_j\right)\mathbf1.
\]
Note that these two terms have nonnegative inner product since
\[
 \left\langle
 (a_j)_{j=1}^d,\,
 \frac{s}{T-ds}\left(\sum_{\ell=1}^d a_\ell\right)\mathbf1
 \right\rangle
 =
 \frac{s}{T-ds}
 \left(\sum_{\ell=1}^d a_\ell\right)
 \left(\sum_{j=1}^d a_j\right)
 =
 \frac{s}{T-ds}\left(\sum_{j=1}^d a_j\right)^2
 \ge0,
\]
and therefore, for every $s\in S$
and $0\le k\le T$, success on $I_s$ implies
\[
 \sum_{j=1}^d D_j\left(\frac{s}{N},\frac{k}{T}\right)^2
 \le \|\Res_k(v_\pi)\|_2^2
 \le \rho^2.
\]
We now define the set $E=\bigcup_{s\in S}((s-1)/N,s/N)$, so that each successful
stopping time contributes an interval of length $1/N$ and
$|E|=|S|/N=p$. We extend the discrepancy bound at successful
grid points to these intervals to show that
\begin{equation}
 \int_E\int_0^1\sum_{j=1}^d D_j(x,y)^2\,dy\,dx
 \le p\left(\rho+\frac{3\sqrt d}{2}\right)^2.
 \label{eq:dimension-success-integral}
\end{equation}
For any $1\le s\le N$, $1\le k\le T$, and
$(x,y)\in((s-1)/N,s/N)\times((k-1)/T,k/T)$, we have
\[
 \frac{i-1}{N}<x\iff i\le s,
 \qquad
 \frac{p_{i,j}-1}{T}<y\iff p_{i,j}\le k.
\]
Thus $(x,y)$ and the upper corner $(s/N,k/T)$ count exactly
the same points, so only the area term changes:
\[
 0\le D_j(x,y)-D_j\left(\frac{s}{N},\frac{k}{T}\right)
 =\frac{sk}{T}-Nxy =N\left(\frac{s}{N}-x\right)y
   +s\left(\frac{k}{T}-y\right)
 \le1+\frac{s}{T}\le\frac32.
\]
When $s\in S$, success bounds the discrepancy norm at the grid
corner by $\rho$, and the triangle inequality gives
\[
 \left(\sum_{j=1}^d D_j(x,y)^2\right)^{1/2}
 \le
 \left(\sum_{j=1}^d
 D_j\left(\frac{s}{N},\frac{k}{T}\right)^2\right)^{1/2}
 +\frac{3\sqrt d}{2}
 \le\rho+\frac{3\sqrt d}{2}.
\]
This holds almost everywhere on $E\times[0,1]$; squaring and
integrating over this set of area $p$ proves
\eqref{eq:dimension-success-integral}.

In the following claim, we prove a lower bound of the contribution of each $D_j$ to the integral \eqref{eq:dimension-success-integral}, i.e., a lower bound of $\int_E\int_0^1 D_j(x,y)^2\,dy\,dx$ for any each coordinate $j$. 
\begin{claim}\label{clm:restricted-roth}
Let $N\ge1$ and let $(u_i,v_i)\in[0,1]^2$ for $i=1,\ldots,N$.
For $(x,y)\in[0,1]^2$, define
\[
 D(x,y):=
 \left|\left\{i\in[N]:u_i<x,\ v_i<y\right\}\right|-Nxy.
\]
Then for every $E\subseteq[0,1]$ which is a union of intervals, with $p=|E|>0$, 
\[
 \int_E\int_0^1 D(x,y)^2\,dy\,dx
 \ge Cp^3\log(N+2),
\]
where $C>0$ is an absolute constant.
\end{claim}
\begin{proof}[Proof of \Cref{clm:restricted-roth}]
We apply Roth's empty-rectangle argument~\cite{Roth54} on $E\times[0,1]$, letting $n=\lceil\log_2(2N)\rceil$ for succinctness.
First we construct two recursive partitions of $[0,1]$, one for each coordinate.
\begin{itemize}
    \item For the $x$-coordinate, start with $[0,1]$. Split each interval
    $I$ at a point $z$ into a left child $I_-$ and a right child
    $I_+$ such that
    \[
     |E\cap I_-|=|E\cap I_+|=\frac{|E\cap I|}{2}.
    \]
    Thus the split divides the length of the part of $E$ inside
    $I$ equally between the children; the children themselves
    need not have equal lengths. Such a point $z$ exists because
    the length of $E$ to its left varies continuously with $z$.
    Let $\mathcal I_a$ denote the collection of intervals after
    $a$ rounds of splitting. It contains $2^a$ intervals, each
    satisfying $|E\cap I|=p2^{-a}$.
    \item For the $y$-coordinate, repeatedly split each interval at its
    midpoint. Let $\mathcal J_b$ denote the collection after $b$
    rounds. It contains $2^b$ intervals, each of length $2^{-b}$.
    We construct both partitions through level $n+1$, so that
    every interval used below also has two children.
\end{itemize}
For each integer $a=0,\ldots,n$, consider all rectangles
$I\times J$ with $I\in\mathcal I_a$ and
$J\in\mathcal J_{n-a}$, which satisfy $|E\cap I|=p2^{-a}, |J|=2^{-(n-a)}$. There are $2^a2^{n-a}=2^n$ such rectangles, with disjoint
interiors, and, since the point set defining $D$ has $N$ points,
at most $N\le2^{n-1}$ rectangles contain a point in their
interior. We select exactly $2^{n-1}$ rectangles whose interiors
contain no points, which over all $a$, selects
$(n+1)2^{n-1}$ rectangles. Next, for each selected rectangle $I\times J$, let us write $I_-,I_+$
for the left and right children of $I$, and $J_-,J_+$ for
the lower and upper halves of $J$. We then define
\[
 h_{I,J}(x,y)=
 \begin{cases}
  +1,&(x,y)\in(I_-\times J_-)\cup(I_+\times J_+),\\
  -1,&(x,y)\in(I_-\times J_+)\cup(I_+\times J_-),\\
  0,&(x,y)\notin I\times J,
 \end{cases}
\]
with values on the boundaries that may be set to zero. 
Applying the Cauchy--Schwarz inequality to $D$
and $\sum_{I,J}h_{I,J}$ gives
\begin{equation}\label{eq:roth-cauchy-schwarz}
    \int_E\int_0^1 D(x,y)^2\,dy\,dx
    \ge
    \frac{
    \left(\int_E\int_0^1
    D(x,y)\sum_{I,J}h_{I,J}(x,y)\,dy\,dx\right)^2
    }{
    \int_E\int_0^1
    \left(\sum_{I,J}h_{I,J}(x,y)\right)^2\,dy\,dx
    },
\end{equation}
where the sum is over all selected rectangles $(I,J)$. We will deal with the two terms denominator and numerator separately. First we compute the denominator: since the lengths of $E\cap I_-$ and $E\cap I_+$ are equal and so are $J_-$ and $J_+$, for every fixed $y$ and every fixed $x$, we have
\[
 \int_E h_{I,J}(x,y)\,dx=
 \int_0^1 h_{I,J}(x,y)\,dy=0, 
 \qquad
 \int_E\int_0^1 h_{I,J}(x,y)^2\,dy\,dx
 =|E\cap I|\,|J|
 =p2^{-n}.
\]
The functions $h_{I,J}$ associated with the selected rectangles
are pairwise orthogonal, since either the rectangles have disjoint
interiors, or the finer $x$-sign has integral zero over $E$
while the coarser sign is constant on that interval. Hence
\begin{equation}
 \int_E\int_0^1
 \left(\sum_{I,J}h_{I,J}(x,y)\right)^2\,dy\,dx
 =(n+1)2^{n-1}p2^{-n}
 =\frac{p(n+1)}{2}.
 \label{eq:roth-denominator}
\end{equation}
We next bound the numerator: for an empty rectangle, every point has at least one
coordinate outside the corresponding interval's interior.
Its counting indicator is constant almost everywhere in
that coordinate, so its integral against $h_{I,J}$ vanishes.
Consequently, only the term $-Nxy$ in $D$ contributes:
\[
 -\int_E\int_0^1 D(x,y)h_{I,J}(x,y)\,dy\,dx
 =
 N\int_{E\cap I}\int_J xyh_{I,J}(x,y)\,dy\,dx.
\]
To bound this expression, let $z$ be the splitting point of
$I$ and put $q=p2^{-a-1}$, the common length of
$E\cap I_-$ and $E\cap I_+$. Then
\begin{align*}
 \int_{E\cap I_+}x\,dx-\int_{E\cap I_-}x\,dx
 =
 \int_{E\cap I_+}(x-z)\,dx
 +\int_{E\cap I_-}(z-x)\,dx\ge 2\int_0^q t\,dt
 =q^2
 =\frac{p^2 2^{-2a}}{4}.
\end{align*}
The inequality holds because, among sets of length $q$
on either side of $z$, the interval adjacent to $z$
minimizes the integral of the distance from $z$.
For the two equal halves of $J$,
\[
 \int_{J_+}y\,dy-\int_{J_-}y\,dy=\frac{|J|^2}{4}.
\]
The definition of $h_{I,J}$ therefore gives
\begin{align*}
 -\int_E\int_0^1 D(x,y)h_{I,J}(x,y)\,dy\,dx
 &=
 N\left(
 \int_{E\cap I_+}x\,dx-\int_{E\cap I_-}x\,dx
 \right)
 \left(
 \int_{J_+}y\,dy-\int_{J_-}y\,dy
 \right)\\
 &\ge \frac{Np^2 2^{-2a}|J|^2}{16}
 =\frac{Np^2 2^{-2n}}{16}.
\end{align*}
Summing the preceding bounds and using
$2^n<4N$, we obtain
\begin{align*}
 -\int_E\int_0^1
 D(x,y)\sum_{I,J}h_{I,J}(x,y)\,dy\,dx
 \ge
 (n+1)2^{n-1}\frac{Np^2 2^{-2n}}{16}=\frac{Np^2(n+1)}{32\cdot2^n}
 \ge\frac{p^2(n+1)}{128}.
\end{align*}
The integral is negative, so squaring gives
\begin{equation}
 \left(\int_E\int_0^1
 D(x,y)\sum_{I,J}h_{I,J}(x,y)\,dy\,dx\right)^2
 \ge \frac{p^4(n+1)^2}{128^2}.
 \label{eq:roth-numerator}
\end{equation}
Substituting \eqref{eq:roth-denominator} and
\eqref{eq:roth-numerator} into \eqref{eq:roth-cauchy-schwarz} gives, for an absolute constant $C>0$,
\[
 \int_E\int_0^1 D(x,y)^2\,dy\,dx
 \ge Cp^3(n+1)
 \ge Cp^3\log(N+2),
\]
as required.
\end{proof}
We are now ready to conclude the proof of \Cref{lem:random-dimension-lower}. First apply \cref{clm:restricted-roth} to the points defining $D_j$ and $E$, we have
\begin{equation}
 \int_E\int_0^1 D_j(x,y)^2\,dy\,dx
 \ge Cp^3\log(N+2).
 \label{eq:coordinate-discrepancy-lower}
\end{equation}
We then sum up \eqref{eq:coordinate-discrepancy-lower} over $j$ and apply \eqref{eq:dimension-success-integral} to get
\[
 p\left(\rho+\frac{3\sqrt d}{2}\right)^2
 \ge Cdp^3\log(N+2).
\]
The one-cycle input forces $\rho\ge\sqrt d/2$, since its cost
is at least $\sqrt d/2$ under every ordering and the success
probability is positive. Hence $\rho+3\sqrt d/2\le4\rho$.
Using $p\ge1-\delta$ and $N=\lfloor T/(2d)\rfloor$, we obtain
\[
 \rho\ge c_0p\sqrt{d\log\left(2+\frac{T}{d}\right)}
 \ge c_0(1-\delta)\sqrt{d\log\left(2+\frac{T}{d}\right)}
\]
for an absolute constant $c_0>0$.
Eventually, we may choose $c\le c_0$ sufficiently small to include the case
$d>T/2$ discussed at the very beginning, which will complete the proof.
\end{proof}

\begin{proof}[Proof of \cref{thm:random-lower-overview}]
The deterministic assertions follow from
\cref{lem:deterministic-lower}.
For $0<\delta\le1/2$, a uniform randomized guarantee must satisfy
both the scalar lower bound of \cref{lem:random-scalar-lower},
embedded in the first coordinate, and the dimension lower bound
of \cref{lem:random-dimension-lower}.
Since $1-\delta\ge1/2$ and the maximum of two nonnegative
quantities is at least half their sum, these bounds give
\[
 \rho\ge
 c\left(
   \log T+\sqrt{d_*\log(2+T/d_*)}
 \right),
\]
for an absolute constant $c>0$, as claimed.
\end{proof}